\documentclass[format=acmsmall]{acmart} 

\setcopyright{none}
\acmDOI{}

\usepackage{array,xspace,multirow,hhline,graphicx,xcolor,tikz,colortbl,tabularx,booktabs,fixltx2e}
\usepackage{cleveref}
\usepackage{eucal}
\usepackage{nicefrac}
\usepackage[multiple]{footmisc}

\usepackage{varwidth}
\usepackage{longtable}
\usepackage{csquotes} 
\usepackage{nccmath}
\usepackage{mathtools}

\newcommand{\pref}{\succsim}
\newcommand{\sPref}{\succ}
\newcommand{\indiff}{\sim}

\newcommand{\replace}[3]{#1^{#2\mapsto {#3}}}

\newcommand{\sd}[1][]{\ensuremath{\ifthenelse{\equal{#1}{}}{\mathit{SD}}{\mathit{SD}(#1)}}}
\newcommand{\rsd}[1][]{\ifthenelse{\equal{#1}{}}{\mathit{RSD}}{\mathit{RSD}(#1)}}	
\newcommand{\supp}[1]{\mathrm{supp}(#1)}

\newcommand{\can}{\mathfrak{c}}

\newcommand{\isabellehol}{\textsf{Isabelle\slash HOL}\xspace}
\renewcommand{\enquote}[1]{``#1''}

\newcommand{\tabref}[1]{Table~\ref{#1}}

\newcommand{\appref}[1]{Appendix~\ref{#1}}

\newcommand{\exref}[1]{Example~\ref{#1}}

\newcommand{\ie}{i.e.,\xspace}
\newcommand{\eg}{e.g.,\xspace}
\newcommand{\cf}{cf.\@\xspace}

\acmJournal{JACM}
\title{Proving the Incompatibility of Efficiency and Strategyproofness via {SMT} Solving}

\author{Florian Brandl}
\email{brandlfl@in.tum.de}
\affiliation{Technische Universit\"at M\"unchen}
\author{Felix Brandt}
\email{brandtf@in.tum.de}
\affiliation{Technische Universit\"at M\"unchen}
\author{Manuel Eberl}
\email{eberlm@in.tum.de}
\affiliation{Technische Universit\"at M\"unchen}
\author{Christian Geist}
\email{geist@in.tum.de}
\affiliation{Technische Universit\"at M\"unchen}

\begin{document}

\begin{abstract}
Two important requirements when aggregating the preferences of multiple agents are that the outcome should be economically efficient and the aggregation mechanism should not be manipulable. 
In this paper, we provide a computer-aided proof of a sweeping impossibility using these two conditions for randomized aggregation mechanisms. 
More precisely, we show that every efficient aggregation mechanism can be manipulated for \emph{all} expected utility representations of the agents' preferences. 
This settles an open problem 
and strengthens a number of existing theorems, including statements that were shown within the special domain of assignment. 
Our proof is obtained by formulating the claim as a satisfiability problem over predicates from real-valued arithmetic, which is then checked using an SMT (satisfiability modulo theories) solver. In order to verify the correctness of the result, a minimal unsatisfiable set of constraints returned by the SMT solver was translated back into a proof in higher-order logic, which was automatically verified by an interactive theorem prover. 
To the best of our knowledge, this is the first application of SMT solvers in computational social choice.
\end{abstract}

\maketitle

\section{Introduction}

Models and results from microeconomic theory, in particular from game theory and social choice, have proven to be very valuable when reasoning about computational multiagent systems \citep[see, \eg][]{NRTV07a,ShLe08a,Roth15a,BCE+14a}.
Two fundamental notions in this context are efficiency---no agent can be made better off without making another one worse off---and strategyproofness---no agent can obtain a more preferred outcome by manipulating his preferences.
\citet{Gibb73a} and \citet{Satt75a} have shown that every strategyproof social choice function is either dictatorial or imposing. Hence, strategyproofness can only be achieved at the cost of discriminating among the agents or among the alternatives. One natural possibility to restore fairness is to allow for randomization. 
Functions that map a profile of individual preferences to a probability distribution over alternatives (a so-called \emph{lottery}) are known as \emph{social decision schemes (SDSs)}. The use of lotteries for the selection of officials interestingly goes back to the world's first democracy in Athens, where it was widely regarded as a principal characteristic of democracy \citep{Head33a}, and has recently gained increasing attention in political science \citep[see, \eg][]{Dowl09a,Guer14a} and social choice \citep[see, \eg][]{Bran17a}.

Generalizing his previous result, \citet{Gibb77a} proved that the only strategyproof and \emph{ex post} efficient social decision schemes are randomizations over dictatorships.\footnote{Alternative proofs for this important theorem were provided by \citet{Dugg96a}, \citet{Nand97a}, and \citet{Tana03a}.} Gibbard's notion of strategyproofness requires that no agent is better off by manipulating his preferences for \emph{some} expected utility representation of the agents' ordinal preferences. This condition is quite demanding because an SDS may be deemed manipulable just because it can be manipulated for a contrived and highly unlikely utility representation.
In this paper, we adopt a weaker notion of strategyproofness, first used by \citet{PoSc86a} and popularized by \citet{BoMo01a}. 
This notion requires that no agent should be better off by manipulating his preferences for \emph{all} expected utility representations of the agents' preferences. 
At the same time, we use a stronger notion of efficiency than \citet{Gibb77a}. 
This notion is defined in analogy to our notion of strategyproofness and requires that no agent can be made better off for \emph{all} utility representations of the agents' preferences, without making another one worse off for \emph{some} utility representation. 
This type of efficiency was introduced by \citet{BoMo01a} and is also known as ordinal efficiency or $\sd$-efficiency where $\sd$ stands for stochastic dominance.

Our main result establishes that no anonymous and neutral SDS satisfies efficiency and strategyproofness. 
This settles a conjecture by \citet{ABBH12a} and strengthens theorems by \citet{ABBH12a}, \citet{ABB13d}, and \citet{BBS15b}. 
It also generalizes related statements that were shown within the special domain of assignment, when interpreting them as social choice results \citep{Zhou90a,BoMo01a,KaSe06a,Koji09a,Nest16a,AzKa17a}. 

The proof of our main result heavily relies on computer-aided solving techniques. 
These techniques were introduced in computational social choice by \citet{TaLi09a}, who reduce well-known impossibility results, such as Arrow's theorem, to finite instances, which can then be checked by a Boolean satisfiability (SAT) solver. 
More recently, this idea has been adapted to more complex settings and axioms while focussing on proving new results rather than reproducing existing ones \citep{GeEn11a,BBGH15a,BGP15a,BrGe15a}. An overview of computer-aided theorem proving in computational social choice is given by \citet{GePe17a}.

In this paper, we go beyond the SAT-based techniques of previous contributions by designing an SMT (satisfiability modulo theories) encoding that captures axioms for \emph{randomized} social choice. 
SMT can be viewed as an enriched form of the satisfiability problem (SAT) where Boolean variables are replaced by statements from a \emph{theory}, such as specific data types or arithmetics. 
Similar to SAT, there is a range of SMT solvers developed by an active community that runs annual competitions \citep{BDM+13a}.  
Typically, SMT solvers are used as backends for verification tasks such as the verification of software. 
To capture axioms about lotteries, we use the theory of (quantifier-free) linear real arithmetic. 
Solving this version of SMT can be seen as an extension to \emph{linear programming} in which arbitrary Boolean operators are allowed to connect (in-)equalities.%

Following the idea of \citet{BrGe15a}, we extracted a \emph{minimal unsatisfiable set (MUS)} of constraints in order to verify our result. 
Despite its relatively complex 94 (non-trivial) constraints, which operate on 47 canonical preference profiles, the MUS was translated back into a proof in higher-order logic, which in turn was verified via the interactive theorem prover \isabellehol. This releases any need to verify our program for generating the SMT formula. 
We also translated this proof into a human-readable---but tedious to check---proof, which is given in the Appendix.

\section{The Model}\label{sec:prelims}

Let $A$ be a finite set of $m$ alternatives and $N=\{1,\dots,n\}$ a set of agents.
A \emph{(weak) preference relation} is a complete and transitive binary relation on~$A$.
The preference relation reported by agent~$i$ is denoted by~$\pref_i$, and the set of all preference relations by~$\mathcal{R}$.
In accordance with conventional notation, we write~$\sPref_i$ for the strict part of~$\pref_i$, \ie~$x \sPref_i y$ if~$x \pref_i y$ but not~$y \pref_i x$, and~$\indiff_i$ for the indifference part of~$\pref_i$, \ie~$x \indiff_i y$ if~$x \pref_i  y$ and~$y \pref_i  x$. 
A preference relation $\pref_i$ is \emph{linear} if $x \sPref_i y$ or~$y\sPref_i x$ for all distinct alternatives $x,y \in A$. 
We will compactly represent a preference relation as a comma-separated list where all alternatives among which an agent is indifferent are placed in a set. For example, $x \sPref_i y \indiff_i z$ is represented by 
${\pref_i} \colon x, \{y,z\}.$
A \emph{preference profile} $R = ({\pref_1},\dots, {\pref_n})$ is an $n$-tuple containing a preference relation $\pref_i$ for each agent $i \in N$. The set of all preference profiles is thus given by $\mathcal{R}^N$. For a given $R\in \mathcal{R}^N$ and ${\succsim}\in \mathcal{R}$, $\replace{R}{i}{\succsim}$ denotes a preference profile identical to $R$ except that $\succsim_i$ is replaced with $\succsim$, \ie $\replace{R}{i}{\succsim} = R\setminus\{(i,{\pref}_i)\}\cup\{(i,\succsim)\}\text.$

\subsection{Social Decision Schemes}

Our central objects of study are social decision schemes: functions that map a preference profile to a \emph{lottery (or probability distribution)} over the alternatives. 
The set of all lotteries over $A$ is denoted by $\Delta(A)$, \ie $\Delta(A) = \{p\in \mathbb{R}_{\ge 0}^A\colon \sum_{x\in A} p(x) = 1\}$, where $p(x)$ is the probability that $p$ assigns to $x$.
Formally, a \emph{social decision scheme (SDS)} is a function $f \colon \mathcal{R}^N \rightarrow \Delta(A)$. By $\supp p$ we denote the \emph{support} of a lottery $p\in\Delta(A)$, \ie the set of all alternatives to which $p$ assigns positive probability.
Two common minimal fairness conditions for SDSs are anonymity and neutrality, \ie symmetry with respect to agents and alternatives, respectively. 
Formally, \emph{anonymity} requires that $f(R)=f(R\circ\sigma)$ for all $R\in \mathcal{R}^N$ and permutations $\sigma\colon N\rightarrow N$ over agents. 
\emph{Neutrality}, on the other hand, is defined via permutations over alternatives. 
An SDS $f$ is \emph{neutral} if $f(R)(x) = f(\pi(R))(\pi(x))$ for all $R \in \mathcal{R}^N$, permutations $\pi\colon A\rightarrow A$, and $x \in A$.\footnote{$\pi(R)$ is the preference profile obtained from $\pi$ by replacing $\pref_i$ with $\pref_i^\pi$ for every $i\in N$, where $\pi(x) \mathrel{\pref_i^{\pi}} \pi(y)$ if and only if $x \pref_i y$.}

\subsection{Efficiency and Strategyproofness}\label{sec:effAndSP}

Many important properties of SDSs, such as efficiency and strategyproofness, require us to reason about the preferences that agents have over lotteries. 
This is commonly achieved by assuming that in a preference profile $R$ every agent $i$, in addition to this preference relation $\succsim_i$, is equipped with a von Neumann-Morgenstern (vNM) \emph{utility function} $u^{R}_i\colon A\rightarrow \mathbb{R}$. 
By definition, a utility function $u^{R}_i$ has to be consistent with the ordinal preferences, \ie for all $x,y\in A$, $u^{R}_i(x)\ge u^{R}_i(y)$ iff $x \succsim_i y$. 
A \emph{utility representation} $u$ then associates with each preference profile $R$ an $n$-tuple $(u_1^R,\dots,u_n^R)$ of such utility functions. 
Whenever the preference profile $R$ is clear from the context, the superscript will be omitted and we write $u_i$ instead of the more cumbersome $u^{R}_i$. 

Given a utility function $u_i$, agent $i$ prefers lottery $p$ to lottery $q$ iff the expected utility for $p$ is at least as high as that of $q$. With slight abuse of notation the domain of utility functions can be extended to $\Delta(A)$ by taking expectations, \ie
\[
u_i(p) = \sum_{x\in A} p(x)u_i(x)\text.
\]
It is straightforward to define efficiency and strategyproofness using expected utility. 
For a given utility representation $u$ and a preference profile $R$, a lottery $p$ \emph{$u$-(Pareto-)dominates} another lottery $q$ at $R$ if
\begin{align*}
u^R_i(p)&\ge u^R_i(q) \text{ for all }i\in N\text{, and }\\
u^R_i(p)&> u^R_i(q) \text{ for some }i\in N\text.
\end{align*}
An SDS $f$ is \emph{$u$-efficient} if it never returns $u$-dominated lotteries, \ie for all $R\in\mathcal{R}^N$, $f(R)$ is not $u$-dominated at $R$. 
The notion of $u$-strategyproofness can be defined analogously: for a given utility representation $u$, preference profile $R$, agent $i$, and preference relation $\succsim$, an SDS $f$ can be \emph{$u$-manipulated} at $R$ by agent $i$ reporting $\succsim$ if
\[
u^{R}_i(f(\replace{R}{i}{\succsim})) > u^{R}_i(f(R))\text.
\]
An SDS is \emph{$u$-strategyproof} if there is no preference profile $R$, agent $i$, and preference relation $\succsim$ such that it can be $u$-manipulated at $R$ by agent $i$ reporting $\succsim$.

The assumption that the vNM utility functions of all agents (and thus their complete preferences \emph{over lotteries}) are known is quite unrealistic. 
Often even the agents themselves are uncertain about their preferences over lotteries and are only aware of their ordinal preferences over alternatives.\footnote{When assuming that all agents possess vNM utility functions, these utility functions could be taken as inputs for the aggregation function. Such aggregation functions are called \emph{cardinal decision schemes} \citep[see, \eg][]{DPS07a}. Concrete vNM utility functions are often unavailable and their representation may require infinite space.} 
A natural way to model this uncertainty is to leave the utility functions unspecified and instead \emph{quantify over all utility functions} that are consistent with the agents' ordinal preferences. 
This modeling assumption leads to much weaker notions of efficiency and strategyproofness. 
\begin{definition} 
An SDS is \emph{efficient} if it never returns a lottery that is $u$-dominated for all utility representations $u$.
\end{definition}
As mentioned in the introduction, this notion of efficiency is also known as \emph{ordinal efficiency} or \emph{$\sd$-efficiency} \citep[see, \eg][]{BoMo01a,ABB14b,ABBB15a}. 
The relationship to stochastic dominance will be discussed in more detail in \Cref{sec:sd}. 

\begin{example}\label{ex:eff}
Consider $A=\{a,b,c,d\}$ and the preference profile $R=({\pref_1},\dots,{\pref_4})$,
\begin{align*}
	&{\pref_1}\colon \{a,c\},\{b,d\}\text, &{\pref_2}\colon& \{b,d\},\{a,c\}\text,\\&{\pref_3}\colon \{a,d\},b,c\text, &{\pref_4}\colon& \{b,c\},a,d
\end{align*} 
Observe that the lottery $\nicefrac{7}{24}\, a+\nicefrac{7}{24}\, b+\nicefrac{5}{24}\, c+\nicefrac{5}{24}\, d$, which is returned by the well-known SDS \emph{random serial dictatorship ($\rsd$)}, is $u$-dominated by $\nicefrac{1}{2}\, a+\nicefrac{1}{2}\, b$ for every utility representation $u$. 
Hence, any SDS that returns this lottery for the profile $R$ would not be efficient. 
On the other hand, the lottery $\nicefrac{1}{2}\, a+\nicefrac{1}{2}\, b$ is not $u$-dominated, which can, for instance, be checked via linear programming (see \Cref{lem:lp}). 
\end{example}

We can also define a weak notion of strategyproofness in analogy to our notion of efficiency.
\begin{definition}
An SDS is \emph{manipulable} if there is a preference profile $R$, an agent $i$, and a preference relation $\succsim$ such that it is $u$-manipulable at $R$ by agent $i$ reporting $\succsim$ for all utility representations $u$.
An SDS is \emph{strategyproof} if it is not manipulable.
\end{definition}
Alternatively, there is a stronger version of strategyproofness first considered by \citet{Gibb77a}, which prescribes that an SDS should be $u$-strategyproof for all utility representations $u$.

For more information concerning the relationship between sets of possible utility functions and preference extensions, such as stochastic dominance, the reader is referred to \citet{ABB14b}.

\section{The Result}\label{sec:result}

Our main result shows that efficiency and strategyproofness are incompatible with basic fairness properties. \citet{ABBH12a} raised the question whether there exists an anonymous, efficient, and strategyproof SDS. When additionally requiring neutrality, we can answer this question in the negative. 

\begin{theorem}\label{thm:sdsConj}\label{THM:SDSCONJ} 
	If $m\ge 4$ and $n\ge 4$, there is no anonymous and neutral SDS that satisfies efficiency and strategyproofness.
\end{theorem}

The proof of \Cref{thm:sdsConj}, which heavily relies on computer-aided solving techniques, is discussed in \Cref{sec:proof}. 
Let us first discuss the independence of the axioms and relate the result to existing theorems. 
$\rsd$ satisfies all axioms \emph{except efficiency}; another SDS known as \emph{maximal lotteries} satisfies all axioms \emph{except strategyproofness} \citep[cf.][]{ABBB15a}. 
Serial dictatorship, the deterministic version of $\rsd$, satisfies neutrality, efficiency, and strategyproofness \emph{but violates anonymity}. 
It is unknown whether \Cref{thm:sdsConj} still holds when dropping the assumption of neutrality.
Our proof, however, only requires a technical weakening of neutrality (cf. \Cref{sec:framework}). 

\subsection{Related Results for Social Choice}

\Cref{thm:sdsConj} generalizes several existing results and is closely related to a number of results in subdomains of social choice. 
\citet{ABBH12a} proved a weak version of \Cref{thm:sdsConj} for the rather restricted class of majoritarian SDSs, \ie SDSs whose outcome may only depend on the pairwise majority relation. 
This statement has later been generalized by \citet{ABB13d} to all SDSs whose outcome only depends on the \emph{weighted} majority relation. 
More recently, \citet{BBS15b} have shown that while random dictatorship satisfies efficiency and strategyproofness (as well as anonymity and neutrality) on the domain of linear preferences, it cannot be extended to the full domain of weak preferences without violating at least one of these properties. 
Their theorem is a direct consequence of \Cref{thm:sdsConj}. 
Other impossibility results have been obtained for stronger notions of efficiency and strategyproofness, which weakens the corresponding statements. 
\citet{ABB13d} have shown that there is no anonymous and neutral SDS that satisfies efficiency and strategyproofness with respect to the \emph{pairwise comparison} lottery extension and with respect to the \emph{upward lexicographic} extension.\footnote{The statement for the pairwise comparison extension holds for at least three agents and three alternatives, whereas \Cref{thm:sdsConj} does not hold for less then four alternatives since $\rsd$ satisfies all properties for up to three alternatives. In contrast to \Cref{thm:sdsConj}, the statement for the upward lexicographic extension does not require neutrality and also holds for linear preferences.}
Both of these notions of efficiency and strategyproofness are stronger than the ones used in \Cref{thm:sdsConj}.

\subsection{Related Results for Assignment}\label{sec:assign}
A subdomain of social choice that has been thoroughly studied in the literature is the assignment (aka house allocation or two-sided matching with one-sided preferences) domain. 
Assignment problems are concerned with the allocation of objects to agents based on the agents' preferences over objects. 
An assignment problem can be associated with a social choice problem by letting the set of alternatives be the set of deterministic allocations and postulating that agents are indifferent among all allocations in which they receive the same object \citep[see, \eg][]{ABB13b}.\footnote{Note that this transformation turns assignment problems with linear preferences over $k$ objects into social choice problems with weak preferences over $k!$ allocations.} Thus, impossibility results for the assignment setting can be interpreted as impossibility results for the social choice setting because they even hold in a smaller domain and an SDS that satisfies efficiency and strategyproofness in the social choice domain also satisfies these properties in any subdomain.

In the following we discuss impossibility results in the assignment domain which, if interpreted for the social choice domain and when assuming anonymity and neutrality, can be seen as weaker versions of \Cref{thm:sdsConj} because they are based on stronger notions of efficiency or strategyproofness or require additional properties. 
In a very influential paper, \citet{BoMo01a} have shown that no randomized assignment mechanism satisfies both efficiency and a strong notion of strategyproofness while treating all agents equally. 
The underlying notion of strategyproofness is identical to the one used by \citet{Gibb77a} and prescribes that the SDS cannot be $u$-manipulated for \emph{any} utility representation $u$. 
The result by \citeauthor{BoMo01a} even holds when preferences over objects are single-peaked \citep{Kasa12a}.
In a related paper, \citet{KaSe06a} proved that no assignment mechanism satisfies efficiency, strategyproofness, and envy-freeness 
for the full domain of preferences.\footnote{Envy-freeness is a fairness property that is stronger than \emph{equal treatment of equals} as used by \citet{BoMo01a}.}
Related impossibility theorems for varying notions of envy-freeness and for multi-unit demand with additive preferences were shown by \citet{Nest16a}, \citet{Koji09a}, and \citet{AzKa17a}.

Settling a conjecture by \citet{Gale87a}, \citet{Zhou90a} showed that no cardinal assignment mechanism satisfies $u$-efficiency and $u$-strategyproofness while treating all agents equally.\footnote{The theorem by \citeauthor{Zhou90a} only requires that agents with the same utility function receive the same amount of expected utility but not necessarily the same assignment. \citeauthor{Gale87a}'s original conjecture assumed equal treatment of equals.} 
The relationship between \citeauthor{Zhou90a}'s result and \Cref{thm:sdsConj} is not obvious because \citeauthor{Zhou90a}'s theorem concerns cardinal mechanisms, \ie functions that take utility profiles rather than preference profiles as input. 
However, every cardinal assignment mechanism can be associated with an ordinal assignment mechanism by selecting some consistent utility function for every preference relation and returning the outcomes for the corresponding utility profiles.
This transformation turns a $u$-efficient and $u$-strategyproof cardinal mechanism into an efficient and strategyproof ordinal mechanism as these properties are purely ordinal.
Hence, \Cref{thm:sdsConj} implies that there is no anonymous, neutral, $u$-efficient, and $u$-strategyproof cardinal decision scheme.

\section{Proving the Result}\label{sec:proof}
In this section, we first reduce the statement of \Cref{thm:sdsConj} to the special case of $m = 4$ and $n = 4$, which we then prove via SMT solving. 
We present an encoding for any finite instance of \Cref{thm:sdsConj} as an SMT problem in the logic of (quantifier-free) linear real arithmetic (\texttt{QF\_LRA}). 
For compatibility with different SMT solvers our encoding adheres to the \textsf{SMT-LIB} standard \citep{BSC10a}. 
In total, we are going to design the following four types of SMT constraints: 
\begin{itemize}
	\item lottery definitions \eqref{cond:lottery},
	\item the orbit condition\footnote{The orbit condition models a part of neutrality.} \eqref{cond:orbit},
	\item strategyproofness \eqref{cond:strategyproofness}, and
	\item efficiency \eqref{cond:eff}.
\end{itemize}
Other conditions such as anonymity are taken care of by the representation of preference profiles. 
	
We then apply an SMT solver to show that this set of constraints for the case of $m = 4$ and $n = 4$ is unsatisfiable, \ie no SDS $f$ with the desired properties exists, and
explain how the output of the solver can be used to obtain a human-verifiable proof of this result. 

Let us start with the reduction lemma before we turn to the concrete encoding in the following subsections.  

\begin{lemma}\label{lem:reduction}
	If there is an anonymous and neutral SDS $f$ that satisfies efficiency and strategyproofness for $|A|=m$ alternatives and $|N|=n$ agents then, for all $m'\le m$ and $n'\le n$, we can also find an SDS $f'$ defined for $m'$ alternatives and $n'$ agents that satisfies the same properties. 
\end{lemma}
\begin{proof}
	Let $f$ be an anonymous and neutral SDS that satisfies efficiency and strategyproofness for $m$ alternatives and $n$ agents. 
	We define a projection $f'$ of $f$ onto $A'\subseteq A, |A'|=m'\le m$ and $N'=\{1,\dots,n'\}\subseteq N, n'\le n$ that satisfies all required properties: 
	
	For every preference profile $R'$ on $A'$ and $N'$, let $f'(R') = f(R)$, where $R$ is defined by the following conditions:
\begin{align}
	\label{eq1} &{\pref}_i\cap (A'\times A') = {\pref}'_i \text{ for all } i\in N'\text,\\ 
	\label{eq2} &x\sPref_i y \text{ for all } x\in A', y\in A\setminus A' \text{ and } i\in N'\text,\\
	\label{eq3} &y\indiff_i z \text{ for all } y,z\in A\setminus A' \text{ and } i\in N'\text{, and}\\
	\label{eq4} &y\indiff_i z \text{ for all } y,z\in A \text{ and } i\in N\setminus N'\text.
\end{align}
	Informally, by \eqref{eq1} agents in $N'$ have the same preferences over alternatives from $A'$ in $R$ and $R'$. 
	Moreover, by \eqref{eq2} they like every alternative in $A'$ strictly better than every alternative not in $A'$ and by \eqref{eq3} they are indifferent between all alternatives not in $A'$. 
	Finally, by \eqref{eq4} all agents in $N\setminus N'$ are completely indifferent. 
	With these conditions, $R$ is uniquely specified given $R'$, and only lotteries $p$ with $\supp{p}\subseteq A'$ are efficient in $R$. 
	Thus, $f'$ is well-defined and it is 
	left to show that $f'$ inherits the relevant properties from $f$.
	The SDS $f'$ is anonymous since $f$ is anonymous and agents in $N$ can only differ by their preferences over $A'$.
	Neutrality follows as $f$ is neutral and all agents are indifferent between all alternatives not in $A'$.
	Efficiency is satisfied by $f'$ since $f$ is efficient and the same set of lotteries is efficient in $R$ and $R'$.
	Finally, $f'$ is strategyproof because $f$ is strategyproof and the outcomes of $f'$ under the two profiles $R'$ and $\replace{(R')}{i}{\pref'}$ are equal to the outcomes of $f$ under the two (extended) profiles $R$ and $\replace{R}{i}{\pref}$, respectively.
\end{proof}

\subsection{Framework, Anonymity, and Neutrality}\label{sec:framework}
For a given number of agents $n$ and set of alternatives $A$, we encode an arbitrary SDS $f \colon \mathcal{R}^N \rightarrow \Delta(A)$ by a set of real-valued variables $p_{R,x}$ with $R\in \mathcal{R}^N$ and $x\in A$. 
Each $p_{R,x}$ then represents the probability with which alternative $x$ is selected for profile $R$, \ie $p_{R,x} = f(R)(x)$. 

This encoding of lotteries leads to the first simple constraints for our SMT encoding, which ensure that for each preference profile $R$ the corresponding variables $p_{R,x}$, $x\in A$ indeed encode a lottery: 
\begin{align}\tag{Lottery}\label{cond:lottery}
	\begin{split}
	\sum_{x\in A} p_{R,x} = 1 & \text{ for all $R\in\mathcal{R}^N$, and }\\
	p_{R,x} \geq 0 & \text{ for all $R\in\mathcal{R}^N$ and $x\in A$.}
	\end{split}
\end{align}

We are now going to argue that, in conjunction with anonymity and neutrality (see \Cref{sec:prelims}), it suffices to consider these constraints for a subset of preference profiles. 
This is because, in contrast to the other axioms, we directly incorporate anonymity and neutrality into the structure of the encoding rather than formulating them as actual constraints.
Similar to the construction involving canonical tournament representations by \citet{BrGe15a}, we model anonymity and neutrality by computing for each preference profile $R\in\mathcal{R}^N$ a \emph{canonical representation} $R_\can\in\mathcal{R}^N$ with respect to these properties. 
In this representation, two preference profiles $R$ and $R'$ are equal (\ie $R_\can=R'_\can$) iff one can be transformed into the other by renaming the agents and alternatives. 
Equivalently, $R_\can=R'_\can$ iff, for every anonymous and neutral SDS $f$, the lotteries $f(R)$ and $f(R')$ are equal (modulo the renaming of the alternatives). 

The SMT constraints and SMT variables are then instantiated only for these canonical representations $\mathcal{R}^N_\can\subseteq\mathcal{R}^N$. 
Apart from enabling an encoding of anonymous and neutral SDSs without any explicit reference to permutations, this also offers a substantial performance gain compared to considering the full domain $\mathcal{R}^N$ of (non-anonymous and non-neutral) preference profiles: the number of preference profiles for $m = 4$ and $n = 4$ is 31,640,625, 
whereas the number of \emph{canonical} preference profiles is merely 60,865. 

Technically, we compute the canonical representation $R_\can$ as follows: 
Let $R=({\pref_1}, \dots, {\pref_n})\in\mathcal{R}^N$ be a preference profile. 
First, we identify $R$ with a function $r: \mathcal{R} \to \mathbb{N}$, which we call \emph{anonymous preference profile}, and which counts the number of agents with a certain preference relation, \ie $r({\pref}) = |\{i\in N \mid {\pref_i} = {\pref}\}|$, thereby ignoring the identity of the agents. 
This representation fully captures anonymity. 

To additionally enforce neutrality, we had to resort to a computationally demanding, naive solution: given $r$, we compute all anonymous preference profiles $\pi(r)$ that can be achieved via a permutation $\pi\colon A \to A$, and, among those profiles, choose the one $\pi_\mathrm{lexmin}(r)$ with lexicographically minimal values (for some fixed ordering of preference relations).
For the canonical representation $R_\can$ we then pick any preference profile $R\in\mathcal{R}^N$ which agrees with $\pi_\mathrm{lexmin}(r)$, for instance, by again using the same fixed ordering of preference relations. 
Fortunately, this approach is still feasible for the small numbers of alternatives with which we are dealing.

While this representation of preference profiles does not completely capture neutrality---the \emph{orbit condition} \citep[see][]{BrGe15a} is missing---this weaker version suffices to prove the impossibility. 
In favor of simpler proofs we, however, include the simple constraints corresponding to a randomized version of the orbit condition. 

In our context, an \emph{orbit} $O$ of a preference profile $R$ is an equivalence class of alternatives. 
Two alternatives $x,y\in A$ are considered equivalent if $\pi(x) = y$ for some permutation $\pi\colon A\to A$ that maps the anonymous preference profile associated with $R$ to itself (\ie $\pi$ is an automorphism of the anonymous preference profile). 
In such a situation, every anonymous and neutral SDS has to assign equal probabilities to $x$ and $y$. 
We hence require that, for each orbit $O\in\mathcal{O}_R$ of a (canonical) profile $R$, the probabilities $p_{R,x}$ are equal for all alternatives $x\in O$. 
As an SMT constraint, this reads 
\begin{equation}\label{cond:orbit}
	p_{R,x} = p_{R,y}\tag{Orbit}
\end{equation}
for all $R\in\mathcal{R}^N_\can$, $O\in\mathcal{O}_{R}$, and $x, y\in O$. 

\begin{example}
Consider the anonymous preference profile $r$ based on $R$ from \Cref{ex:eff} and the permutation $\pi = (a b)(c d)$.
As $\pi(r)=r$ (and since no other non-trivial permutation has this property) the set of orbits of $R$ is $\mathcal{O}_R=\left\{\{a,b\},\{c,d\}\right\}$. 
\end{example}

\subsection{Stochastic Dominance}\label{sec:sd}
In order to avoid quantifying over utility functions, we leverage well-known representations of efficiency and strategyproofness via \emph{stochastic dominance ($\sd$)} \citep[\cf][]{BoMo01a,McLe02a,ABB14b}. 
Lottery $p$ \emph{stochastically dominates} lottery $q$ for an agent $i$ (short: $p\pref_i^{\sd} q$) if for every alternative $x$, $p$ is at least as likely as $q$ to yield an alternative at least as good as $x$. 
Formally, 
\[
	p\pref_i^{\sd} q \text{ iff } \sum_{y\pref_i x} p(y) \ge \sum_{y\pref_i x} q(y) \text{ for all } x\in A\text.
\]
When $p\pref_i^{\sd} q$ and not $q\pref_i^{\sd} p$ we write $p\sPref_i^{\sd} q$. 

As an example, consider the preference relation ${\pref_i}\colon a,b,c$. 
We then have that  
\[
(\nicefrac{2}{3}\, a + \nicefrac{1}{3}\, c) \sPref_i^{\sd} (\nicefrac{1}{3}\, a + \nicefrac{1}{3}\, b + \nicefrac{1}{3}\, c)
\]
while $\nicefrac{2}{3}\, a + \nicefrac{1}{3}\, c$ and $b$ are incomparable based on stochastic dominance. 

\begin{lemma}\label{lem:sdEquiv}
	Let ${\pref_i}\in\mathcal{R}$. 
	A lottery $p$ $\sd$-dominates another lottery $q$ for an agent $i$ iff $u_i(p) \geq u_i(q)$ for every utility function $u_i$ consistent with $\pref_i$. 
	As a consequence, 
	\begin{enumerate}
		\item an SDS~$f$ is efficient iff, for all $R\in\mathcal{R}^N$, there is no lottery $p$ 
	such that $p\pref_i^{\sd} f(R)$ for all $i\in N$ and $p\sPref_i^{\sd} f(R)$ for some $i\in N$, and
		\item an SDS~$f$ is manipulable iff there exist a preference profile $R$, an agent $i$, and a preference relation $\pref$ such that $f(\replace{R}{i}{\pref}) \sPref_i^{\sd} f(R)$. 
	\end{enumerate}
\end{lemma}

\begin{proof}
	For the direction from left to right, assume that $p\pref_i^{\sd} q$.
	Without loss of generality, let $A = \{x_1,\dots,x_m\}$ and $x_j\pref_i x_k$ if and only if $j\le k$ for all $j,k\in \{1,\dots, m\}$.
	Then, by definition, for all $j\in\{1,\dots,m\}$, $\sum_{k=1}^j p(x_k)\ge \sum_{k=1}^jq(x_k)$.
	Let $u_i$ be a utility function consistent with $\pref_i$, \ie $u_i(x_j)\ge u_i(x_k)$ if and only if $j\le k$.
	Then,
	\begin{align*}
		u_i(p) - u_i(q) = \sum_{j=1}^m (p(x_j) - q(x_j))u_i(x_j)
		= \sum_{j=1}^m \underbrace{(u_i(x_j) - u_i(x_{j+1}))}_{\ge 0} \underbrace{\sum_{k=1}^j (p(x_k) - q(x_k))}_{\ge 0}  \ge 0\text,
	\end{align*}
	where $u_i(x_{m+1})$ is set to $0$. Hence, $u_i(p)\ge u_i(q)$.

	For the direction from right to left, assume that $u_i(p)\ge u_i(q)$ for all utility functions $u_i$ consistent with $\pref_i$.
	Assume for contradiction that $p\not\pref_i^{\sd} q$, \ie there is $x\in A$ such that $\sum_{y\pref_i x} q(x) - \sum_{y\pref_i x} p(x) = \epsilon > 0$.
	Let $u_i$ be a utility function consistent with $\pref_i$ such that $u_i(y) \in [1-\nicefrac{\epsilon}{2},1]$ for all $y\pref_i x$ and $u_i(y) \in [0,\nicefrac{\epsilon}{2}]$ for all $x\sPref_i y$.
	Such a $u_i$ exists, since $\epsilon > 0$.
	Then,
	\[
		u_i(q) \ge (1-\nicefrac{\epsilon}{2}) \sum_{y\pref_i x} q(y) \ge \sum_{y\succsim_i x} q(y) - \nicefrac{\epsilon}{2} > \sum_{y\pref_i x} p(y) + \nicefrac{\epsilon}{2}\ge u_i(p)\text,
	\]
	which contradicts the assumption.
\end{proof}

In words, \Cref{lem:sdEquiv} shows that an SDS $f$ is efficient if and only if $f(R)$ is Pareto-efficient with respect to stochastic dominance for all preference profiles $R$. Secondly, $f$ is manipulable if and only if some agent can misrepresent his preferences to obtain a lottery that he prefers to the lottery obtained by sincere voting with respect to stochastic dominance.

\subsubsection{Encoding Strategyproofness}

Starting from the above equivalence, encoding strategyproofness as an SMT constraint is now a much simpler task. 
For each (canonical) preference profile $R\in\mathcal{R}^N_\can$, agent $i\in N$,\footnote{Note that, due to anonymity, it is not necessary to iterate over all agents~$i$. Rather it suffices to pick one agent per unique preference relation contained in $R$.}
and preference relation ${\pref}\in\mathcal{R}$, we encode that the manipulated outcome $f(\replace{R}{i}{\pref})$ is not $\sd$-preferred to the truthful outcome $f(R)$ by agent $i$: 
\begin{align*}
	&\neg\left(f(\replace{R}{i}{\pref}) \sPref_i^{\sd} f(R) \right)\notag\\
	\equiv&\, f(\replace{R}{i}{\pref}) \not\pref_i^{\sd} f(R) \vee f(R) \pref_i^{\sd} f(\replace{R}{i}{\pref})\notag\\
	\equiv& \left(\left(\exists x\in A\right) \sum_{y\pref_i x} f(\replace{R}{i}{\pref})(y) < \sum_{y\pref_i x} f(R)(y)\right) \vee 
	\left(\left(\forall x\in A\right) \sum_{y\pref_i x} f(\replace{R}{i}{\pref})(y) \stackrel{\scriptscriptstyle(*)}{\le} \sum_{y\pref_i x} f(R)(y)\right)\\
		\equiv& \left(\bigvee_{x\in A} \sum_{y\pref_i x} p_{(\replace{R}{i}{\pref})_\can,\pi^{\replace{R}{i}{\pref}}_\can(y)} < \sum_{y\pref_i x} p_{R,y}\right) \vee 
		\left(\bigwedge_{x\in A} \sum_{y\pref_i x} p_{(\replace{R}{i}{\pref})_\can,\pi^{\replace{R}{i}{\pref}}_\can(y)} \stackrel{\scriptscriptstyle(**)}{=} \sum_{y\pref_i x} p_{R,y}\right)\text{,}
		\tag{Strategyproofness}\label{cond:strategyproofness}
\end{align*}
where $\pi^{\replace{R}{i}{\pref}}_\can$ stands for a permutation of alternatives that (together with a potential renaming of alternatives) leads from ${\replace{R}{i}{\pref}}$ to $({\replace{R}{i}{\pref}})_\can$. The inequality $\scriptstyle(*)$ can be replaced by the equality $\scriptstyle(**)$ since the case of at least one strict inequality is captured by the corresponding disjunctive condition one line above.

\subsubsection{Encoding Efficiency} 
While \Cref{lem:sdEquiv} helps to formulate efficiency as an SMT axiom it is not yet sufficient because a quantification over the set of all lotteries $\Delta(A)$ remains. 
In order to get rid of this quantifier, we apply two lemmas by \citet{ABB14b}, for which we include (slightly simplified) proofs in favor of a self-contained presentation.  
The first lemma states that efficiency of a lottery only depends on its support. 
The second lemma shows that deciding whether a lottery is efficient reduces to solving a linear program. 

\begin{lemma}[\citeauthor{ABB14b}, \citeyear{ABB14b}]\label{lem:support}
	Let $R\in\mathcal{R}^N$.  
	A lottery $p\in\Delta(A)$ is efficient iff every lottery $p'\in\Delta(A)$ with $\supp{p'}\subseteq\supp{p}$ is efficient. 
\end{lemma}

\begin{proof}
	We prove the statement by contraposition: if $p'\in\Delta(A)$ is not efficient, then no lottery $p$ with $\supp{p'}\subseteq\supp{p}$ is efficient. 
	If $p'$ is not efficient, there is $q'\in\Delta(A)$ such that $q'$ $u$-dominates $p'$ for all utility representations $u^R$, \ie for all agents $i\in N$ and all utility functions $u_i$ consistent with $\succsim_i$, $u_i(q') - u_i(p')\ge 0$ and $u_{i'}(q') - u_{i'}(p') > 0$ for some agent $i'\in N$ and all utility functions $u_{i'}$ consistent with $\succsim_{i'}$. 
	Let $v = q' - p'\in\mathbb{R}^{A}$. 
	Note that, for all $x\in A$, $v(x) < 0$ implies $x\in\supp{p'}$. 
	Now let $\epsilon > 0$ small enough such that $q = p + \epsilon v\in\Delta(A)$. 
	This is possible because $\supp{p'}\subseteq\supp{p}$. 
	By definition of $q$ and linearity of $u_i$, we have that, for all $i\in N$ and all $u_i$ consistent with $\succsim_i$, $u_i(q) - u_i(p) = \epsilon u_i(v) = \epsilon (u_i(q') - u_i(p'))\ge 0$ and $u_{i'}(q) - u_{i'}(p) > 0$ for all $u_{i'}$ consistent with $\succsim_{i'}$. 
	Thus, $p$ is not efficient, contradicting the assumption.
\end{proof}

\begin{lemma}[\citeauthor{ABB14b}, \citeyear{ABB14b}]\label{lem:lp}
	Whether a lottery $p\in\Delta(A)$ is efficient for a given preference profile $R$ can be computed in polynomial time by solving a linear program. 	
\end{lemma}
\begin{proof}
	Given the equivalence from \Cref{lem:sdEquiv}, a lottery $p$ is easily seen to be efficient iff the optimal objective value of the following linear program is zero (since then there is no lottery $q$ that $\sd$-dominates $p$): 
	\begin{align*}
		\max_{q} & \sum_{i \in N}\sum_{x \in A}\sum_{y \pref_{i} x}  q_{y}-p_{y}\\
		\text{subject to} & \sum_{y \pref_{i} x} q_{y} \ge \sum_{y \pref_{i} x} p_{y} \text{ for all } x \in A \text{, } i \in N\text, \notag\\
		&\sum_{x \in A} q_{x} = 1\text{,\,\,\,} q_x \geq 0 \text{ for all } x \in A \text.\notag\\
	\end{align*}%
\end{proof}

Recall that an SDS is efficient if it never returns a dominated lottery. 
By \Cref{lem:support}, this is equivalent to never returning a lottery with \emph{inefficient support}. 
To capture this, we encode, for each (canonical) preference profile $R\in\mathcal{R}^N_\can$, that the probability for at least one alternative in every (inclusion-minimal) inefficient support $I_{R}\subseteq A$ is zero: 
\begin{equation}\label{cond:eff}
	\bigvee_{x\in I_{R}} p_{{R},x} = 0\text. \tag{Efficiency}
\end{equation}
Given a preference profile $R$ and a support $I_R$, it can be decided in polynomial time whether $I_R$ is inefficient by checking for an arbitrary lottery with support $I_R$ whether it is efficient and then applying \Cref{lem:support,lem:lp}. The set of inclusion-minimal inefficient supports can be found by iterating over all supports. For a small number of alternatives this is feasible even though the number of possible supports is exponential in the number of alternatives.

\subsection{Restricted Domains}
\label{sec:domains}

Since $\rsd$ is known to satisfy both strategyproofness and efficiency when there are only up to $3$~alternatives or only up to $3$~agents and $5$~alternatives \citep{ABBB15a}, 
the search for an impossibility has to start at $m=4$~alternatives and $n=4$~agents. 
For these parameters, an encoding of the full domain, unfortunately, becomes prohibitively large.  
Hence, for $m=4$ and $n=4$, one has to carefully optimize the domain under consideration, on the one hand, to include a sufficient number of profiles for a successful proof, and, on the other hand, not to include too many profiles, which would prevent the solver from terminating within a reasonable amount of time. 

The following incremental strategy was found to be successful. 
We start with a specific profile $R$, 
from which we only consider sequences of potential manipulations as long as (in each step) the manipulated individual preferences are not too distinct from the truthful preferences. 
To this end, we measure the magnitude of manipulations by the Kendall tau distance $\tau$, which counts pairwise disagreements between $R_i$ and $R'_i$ \citep[see also][]{Sato13a}. 
A change in the individual preferences of an agent will be called a \emph{$k$-manipulation} if $\tau(R_i,R'_i)\leq k$. 
Then, for example, strategically swapping two alternatives is a $2$-manipulation, and breaking or introducing a tie between two alternatives is a $1$-manipulation. 

On the domain which starts from the preference profile $R$ given in \exref{ex:eff} and from there allows sequences of $(1, 2, 1, 2)$-manipulations, we were able to prove the result within a few minutes of running-time.\footnote{I.e., first we allow any $1$-manipulation from $R$, then, from every resulting profile, any $2$-manipulation is allowed (not necessarily by the same agent), and so forth. Showing the result on this domain implies a slightly stronger statement where strategyproofness only has to hold for ``small'' lies (of at most Kendall tau distance 2).}\footnote{The SMT solver \textsf{MathSAT} \citep{CGSS13a} terminates quickly within less than 3~minutes with the suggested competition settings, whereas \textsf{z3} \citep{MoBj08a} requires some additional configuration, but then also supports core extraction within the same time frame.} On smaller domains (\eg when considering $(1, 2, 2)$-manipulations from $R$), the axioms are still compatible. 

\subsection{Verification of Correctness}

The main drawbacks of the SMT-based proof are that \emph{(i)} one must trust the correctness of the SMT solver, \emph{(ii)} one must trust the correctness of the program that performs the encoding into SMT-LIB, and \emph{(iii)} the proof is unstructured and completely unlike a hand-written mathematical argument, which makes it virtually impossible to be checked by humans.

In order to tackle the first issue, we used \textsf{z3} to generate a \emph{minimal unsatisfiable set (MUS)} of constraints, \ie an inclusion-minimal set of constraints such that this set is still unsatisfiable \citep[see, also,][]{BrGe15a}. A MUS corresponding to \Cref{thm:sdsConj} consists of 94~constraints, not counting the (trivial) lottery definitions. This MUS, annotated with \eg the 47~required canonical preference profiles, is available as part of an arXiv version of this paper \citep{BBG16a}. The unsatisfiability of the MUS has been verified by the solvers \textsf{CVC4}, \textsf{MathSAT}, \textsf{Yices2}, and \textsf{z3}. 

We addressed the second issue by performing several sanity checks such as running solvers on multiple variants of the encoding which represent known theorems. This way, we reproduced (amongst others) the results by \citet{BoMo01a} and \citet{KaSe06a}, as well as the possibility result for $m<4$.

To finally remove any doubt about correctness and simultaneously address the third issue, we translated the MUS into an independent proof, which no longer relies on SMT, within the interactive theorem prover \isabellehol \citep{NPW02a,NiKl14a}.  
\textsf{Isabelle} is a generic interactive theorem prover where \emph{interactive} means that the prover does not find the proof by itself like an automated theorem prover---the user must give it a sequence of steps to follow and the prover's automation fills in the gaps. This allows proofs of more complex theorems that are outside the scope of fully-automated theorem provers. 
The proof of \Cref{thm:sdsConj} in \textsf{Isabelle} is about 400 lines long, but still fairly legible since it consists of many individual small proofs. The fact that \textsf{Isabelle} can automatically simplify inequalities using all facts proven so far actually makes conducting the proof in the system much easier, less tedious, and less error-prone than on paper. Moreover, \emph{all} aspects of the proof---including formal definitions of the social-choice-theoretic concepts, the reduction of the general case to that of $m = 4$ and $n = 4$, the generation of the constraints arising from the 47 canonical preference profiles, and the proof of the inconsistency of these constraints (which corresponds to the SMT proof)---have been verified by \isabellehol.

The trustworthiness of such a proof stems from the fact that all \textsf{Isabelle} proofs are broken down into small logical inference steps, which are checked by \textsf{Isabelle}'s kernel. Since only the kernel can produce new theorems, it is sufficient to trust it to correctly implement these inference steps to trust that any proof it accepts really does hold in the underlying logic. Furthermore, the mere act of breaking down proofs into such small steps exposes many mistakes and forgotten side conditions.

The \textsf{Isabelle} proof is available in the \emph{Archive of Formal Proofs} \citep{Eber16c}, which is a peer-reviewed online repository of \textsf{Isabelle} proofs. For more details on the background in \textsf{Isabelle} and how the proof was obtained from the MUS, we refer to \cite{Eber16a}.  A human-readable version of this proof is given in \appref{app:proof}.

\begin{table}[tb]
\begin{center}
\begin{tabular}{lr}
\toprule
Statement & Number of canonical preference profiles\\ 
\midrule
\Cref{thm:sdsConj} & 47 \\[1ex]
\citet[][Theorem 1]{BBS15b} & 13 \\ 
\citet[][Theorem 3]{ABB13d} & 10 \\ 
\citet[][Theorem 2]{ABB13d} & 7 \\ 
\citet[][Theorem 4]{ABB13d} & 7 \\ 
\citet[][Theorem 1]{ABBH12a} & 5 \\ 
\midrule
\citet[][Theorem 2]{BoMo01a} & 11\\
\citet[][Theorem 1]{Kasa12a} & 9\\ 
\citet[][Theorem 2]{Nest16a} & 8\\
\citet[][Theorem 1]{Nest16a} & 6\\
\citet[][Theorem 1]{Zhou90a} & 5\\
\citet[][Theorem 1]{AzKa17a} & 4\\ 
\citet[][Theorem 2]{AzKa17a} & 3\\ 
\citet[][Theorem 1]{Koji09a} & 3\\ 
\citet[][Section 4]{KaSe06a} & 2\\
\citet[][Theorem 1]{Nest16a} & 2\\ 
\bottomrule
\end{tabular}
\end{center}
\caption{Proof complexity comparison of impossibility statements using efficiency and strategyproofness in terms of the number of canonical preference profiles used in the proof. The statements in the lower part of the table have been proven for the assignment domain.}
\label{tab:profiles}
\end{table}

\section{Conclusion}
In this paper, we have leveraged computer-aided solving techniques to prove a sweeping impossibility for randomized aggregation mechanisms. In particular, we have reduced the statement to a finite propositional formula using linear arithmetic, which was then shown to be unsatisfiable by an SMT solver. A crucial step in the construction of the formula was to find a restricted domain of preference profiles that is not too large yet sufficient for the impossibility to hold.

It seems unlikely that this proof would have been found without the help of computers because manual proofs of significantly weaker statements already turned out to be quite complex (see \tabref{tab:profiles} for a comparison of the proof complexity of related statements).
Nevertheless, now that the theorem has been established, our computer-aided methods may guide the search for related, perhaps even stronger statements that allow for more intuitive proofs and provide more insights into randomized social choice.

Generally speaking, we believe that SMT solving and subsequent verification via \textsf{Isabelle} is applicable to a wide range of problems in social choice and, more generally, in microeconomic theory \citep[see][]{GePe17a}. 
In particular, extending our result to the special domain of assignment (see \Cref{sec:assign}) is desirable as this would strengthen a number of existing theorems. 
Other interesting questions are whether the impossibility still holds when weakening efficiency and strategyproofness even further or when omitting neutrality \citep[see][]{Bran17a}.

\begin{acks}
This material is based upon work supported by the Deutsche Forschungsgemeinschaft under grants {BR~2312/7-2} and {BR~2312/10-1} and the TUM Institute for Advanced Study through a Hans Fischer Senior Fellowship.
The authors also thank Alberto Griggio and Mohammad Mehdi Pourhashem Kallehbasti for guidance on how to most effectively use \textsf{MathSAT} and \textsf{z3}, respectively, and the anonymous reviewers for their helpful comments. 

Results of this paper were presented at the 6th International Workshop on Computational Social Choice
(Toulouse, June 2016) and the 25th International Joint Conference on Artificial Intelligence (New York, July 2016).
\end{acks}

\newpage
\begin{appendix}

\section{Proof of Theorem~\ref{thm:sdsConj}}
\label{app:proof}

\subsection{Main Proof}

We will now give the complete human-readable proof of \Cref{thm:sdsConj}. This proof is essentially a paraphrased version of the formal \isabellehol proof, which is available in the AFP entry \citep{Eber16c}.

Our general approach will be to attempt to \enquote{solve} preference profiles, \ie determine the exact value of $f(R_i)(x)$ (which we write as $p_{i,x}$) for a profile $R_i$ and an alternative $x$. Whenever this is not possible, we try to express $p_{i,x}$ in terms of other $p_{j,y}$ or at least find simple inequalities that the $p_{i,x}$ satisfy. We do this until we have gained enough knowledge about the SDS to derive a contradiction.

A typical step in the proofs will be to pick a strategyproofness condition (which usually consist of several disjunctions) and simplify it with all the knowledge that we have---substituting the $p_{i,x}$ whose values we already know, \eg substituting $p_{i,d} = 1 - p_{i,a}$ if we know that $p_{i,b} = p_{i,c} = 0$. We will use the fact that all $p_{i,x}$ are non-negative and that $\sum_{x\in A} p_{i,x} = 1$ without mentioning it explicitly.

Every step of the proof (\ie \enquote{Condition X simplifies to \ldots} or \enquote{Condition X implies \ldots}) is elementary in the sense that it can by solved automatically by \textsf{Isabelle}---in fact, the proof printed here is often considerably more verbose and with more intermediate steps than would be necessary in \textsf{Isabelle}. Still, for a human, most of these steps will require a few steps of reasoning on paper. We chose not to go into more detail of the individual steps, since it would only have made the proof even longer and less readable.

The proof will reference orbit equations, efficiency conditions, and strategyproofness conditions on the set of 47 preference profiles mentioned before. 
As an aid to the reader, the proof contains tables listing all the knowledge that we currently have about the probabilities of the lottery returned by the hypothetical SDS after every few steps.\\

We start by listing the 47 preference profiles used in the proof by giving the weak rankings of each agent.
\begin{center}
\begingroup
\begin{longtable}{ccccc}
\toprule
\hspace*{2mm}Profile\hspace*{2mm} &\hspace*{4mm} Agent 1 \hspace*{4mm}& \hspace*{4mm}Agent 2\hspace*{4mm} & \hspace*{4mm}Agent 3 \hspace*{4mm}&\hspace*{4mm} Agent 4\hspace*{4mm}\\\midrule
$R_{1}$ & $ \{c, d\}, \{a, b\}$ & $\{b, d\}, a, c$ & $a, b, \{c, d\}$ & $\{a, c\}, \{b, d\}$\\
$R_{2}$ & $ \{a, c\}, \{b, d\}$ & $\{c, d\}, a, b$ & $\{b, d\}, a, c$ & $a, b, \{c, d\}$\\
$R_{3}$ & $ \{a, b\}, \{c, d\}$ & $\{c, d\}, \{a, b\}$ & $d, \{a, b\}, c$ & $c, a, \{b, d\}$\\
$R_{4}$ & $ \{a, b\}, \{c, d\}$ & $\{a, d\}, \{b, c\}$ & $c, \{a, b\}, d$ & $d, c, \{a, b\}$\\
$R_{5}$ & $ \{c, d\}, \{a, b\}$ & $\{a, b\}, \{c, d\}$ & $\{a, c\}, d, b$ & $d, \{a, b\}, c$\\
$R_{6}$ & $ \{a, b\}, \{c, d\}$ & $\{c, d\}, \{a, b\}$ & $\{a, c\}, \{b, d\}$ & $d, b, a, c$\\
$R_{7}$ & $ \{a, b\}, \{c, d\}$ & $\{c, d\}, \{a, b\}$ & $a, c, d, b$ & $d, \{a, b\}, c$\\
$R_{8}$ & $ \{a, b\}, \{c, d\}$ & $\{a, c\}, \{b, d\}$ & $d, \{a, b\}, c$ & $d, c, \{a, b\}$\\
$R_{9}$ & $ \{a, b\}, \{c, d\}$ & $\{a, d\}, c, b$ & $d, c, \{a, b\}$ & $\{a, b, c\}, d$\\
$R_{10}$ & $ \{a, b\}, \{c, d\}$ & $\{c, d\}, \{a, b\}$ & $\{a, c\}, d, b$ & $\{b, d\}, a, c$\\
$R_{11}$ & $ \{a, b\}, \{c, d\}$ & $\{c, d\}, \{a, b\}$ & $d, \{a, b\}, c$ & $c, a, b, d$\\
$R_{12}$ & $ \{c, d\}, \{a, b\}$ & $\{a, b\}, \{c, d\}$ & $\{a, c\}, d, b$ & $\{a, b, d\}, c$\\
$R_{13}$ & $ \{a, c\}, \{b, d\}$ & $\{c, d\}, a, b$ & $\{b, d\}, a, c$ & $a, b, d, c$\\
$R_{14}$ & $ \{a, b\}, \{c, d\}$ & $d, c, \{a, b\}$ & $\{a, b, c\}, d$ & $a, d, c, b$\\
$R_{15}$ & $ \{a, b\}, \{c, d\}$ & $\{c, d\}, \{a, b\}$ & $\{b, d\}, a, c$ & $a, c, d, b$\\
$R_{16}$ & $ \{a, b\}, \{c, d\}$ & $\{c, d\}, \{a, b\}$ & $a, c, d, b$ & $\{a, b, d\}, c$\\
$R_{17}$ & $ \{a, b\}, \{c, d\}$ & $\{c, d\}, \{a, b\}$ & $\{a, c\}, \{b, d\}$ & $d, \{a, b\}, c$\\
$R_{18}$ & $ \{a, b\}, \{c, d\}$ & $\{a, d\}, \{b, c\}$ & $\{a, b, c\}, d$ & $d, c, \{a, b\}$\\
$R_{19}$ & $ \{a, b\}, \{c, d\}$ & $\{c, d\}, \{a, b\}$ & $\{b, d\}, a, c$ & $\{a, c\}, \{b, d\}$\\
$R_{20}$ & $ \{b, d\}, a, c$ & $b, a, \{c, d\}$ & $a, c, \{b, d\}$ & $d, c, \{a, b\}$\\
$R_{21}$ & $ \{a, d\}, c, b$ & $d, c, \{a, b\}$ & $c, \{a, b\}, d$ & $a, b, \{c, d\}$\\
$R_{22}$ & $ \{a, c\}, d, b$ & $d, c, \{a, b\}$ & $d, \{a, b\}, c$ & $a, b, \{c, d\}$\\
$R_{23}$ & $ \{a, b\}, \{c, d\}$ & $\{c, d\}, \{a, b\}$ & $\{a, c\}, \{b, d\}$ & $\{a, b, d\}, c$\\
$R_{24}$ & $ \{c, d\}, \{a, b\}$ & $d, b, a, c$ & $c, a, \{b, d\}$ & $b, a, \{c, d\}$\\
$R_{25}$ & $ \{c, d\}, \{a, b\}$ & $\{b, d\}, a, c$ & $a, b, \{c, d\}$ & $a, c, \{b, d\}$\\
$R_{26}$ & $ \{b, d\}, \{a, c\}$ & $\{c, d\}, \{a, b\}$ & $a, b, \{c, d\}$ & $a, c, \{b, d\}$\\
$R_{27}$ & $ \{a, b\}, \{c, d\}$ & $\{b, d\}, a, c$ & $\{a, c\}, \{b, d\}$ & $\{c, d\}, a, b$\\
$R_{28}$ & $ \{c, d\}, a, b$ & $\{b, d\}, a, c$ & $a, b, \{c, d\}$ & $a, c, \{b, d\}$\\
$R_{29}$ & $ \{a, c\}, d, b$ & $\{b, d\}, a, c$ & $a, b, \{c, d\}$ & $d, c, \{a, b\}$\\
$R_{30}$ & $ \{a, d\}, c, b$ & $d, c, \{a, b\}$ & $c, \{a, b\}, d$ & $\{a, b\}, d, c$\\
$R_{31}$ & $ \{b, d\}, a, c$ & $\{a, c\}, d, b$ & $c, d, \{a, b\}$ & $\{a, b\}, c, d$\\
$R_{32}$ & $ \{a, c\}, d, b$ & $d, c, \{a, b\}$ & $d, \{a, b\}, c$ & $\{a, b\}, d, c$\\
$R_{33}$ & $ \{c, d\}, \{a, b\}$ & $\{a, c\}, d, b$ & $a, b, \{c, d\}$ & $d, \{a, b\}, c$\\
$R_{34}$ & $ \{a, b\}, \{c, d\}$ & $a, c, d, b$ & $b, \{a, d\}, c$ & $c, d, \{a, b\}$\\
$R_{35}$ & $ \{a, d\}, c, b$ & $a, b, \{c, d\}$ & $\{a, b, c\}, d$ & $d, c, \{a, b\}$\\
$R_{36}$ & $ \{c, d\}, \{a, b\}$ & $\{a, c\}, d, b$ & $\{b, d\}, a, c$ & $a, b, \{c, d\}$\\
$R_{37}$ & $ \{a, c\}, \{b, d\}$ & $\{b, d\}, \{a, c\}$ & $a, b, \{c, d\}$ & $c, d, \{a, b\}$\\
$R_{38}$ & $ \{c, d\}, a, b$ & $\{b, d\}, a, c$ & $a, b, \{c, d\}$ & $\{a, c\}, b, d$\\
$R_{39}$ & $ \{a, c\}, d, b$ & $\{b, d\}, a, c$ & $a, b, \{c, d\}$ & $\{c, d\}, a, b$\\
$R_{40}$ & $ \{a, d\}, c, b$ & $\{a, b\}, c, d$ & $\{a, b, c\}, d$ & $d, c, \{a, b\}$\\
$R_{41}$ & $ \{a, d\}, c, b$ & $\{a, b\}, d, c$ & $\{a, b, c\}, d$ & $d, c, \{a, b\}$\\
$R_{42}$ & $ \{c, d\}, \{a, b\}$ & $\{a, b\}, \{c, d\}$ & $d, b, a, c$ & $c, a, \{b, d\}$\\
$R_{43}$ & $ \{a, b\}, \{c, d\}$ & $\{c, d\}, \{a, b\}$ & $d, \{a, b\}, c$ & $a, \{c, d\}, b$\\
$R_{44}$ & $ \{c, d\}, \{a, b\}$ & $\{a, c\}, d, b$ & $\{a, b\}, d, c$ & $\{a, b, d\}, c$\\
$R_{45}$ & $ \{a, c\}, d, b$ & $\{b, d\}, a, c$ & $\{a, b\}, c, d$ & $\{c, d\}, b, a$\\
$R_{46}$ & $ \{b, d\}, a, c$ & $d, c, \{a, b\}$ & $\{a, c\}, \{b, d\}$ & $b, a, \{c, d\}$\\
$R_{47}$ & $ \{a, b\}, \{c, d\}$ & $\{a, d\}, c, b$ & $d, c, \{a, b\}$ & $c, \{a, b\}, d$\\\bottomrule
\caption{The 47 preference profiles used in the proof.}
\label{tbl_profiles}
\end{longtable}
\endgroup
\end{center}

Now, to begin with the proof, we shall first focus on those profiles that have rich symmetries (\ie orbit conditions) and restrictive efficiency conditions (\eg by admitting Pareto dominated alternatives).

\Cref{tbl_automorphisms} lists profile automorphisms, \ie permutations of the alternatives such that applying the permutation to the profile yields a profile that is anonymity-equivalent to the original profile. Given such a profile, an anonymous and neutral SDS must return the same probability for all alternatives contained in the same orbit of the permutation.
To increase readability, the permutations are already written as a product of their orbits; for instance, the first orbit condition states that $p_{10,a} = p_{10,d}$ and $p_{10,b} = p_{10,c}$.

\begin{center}
\begin{longtable}{cc}
\toprule
Profile & Permutation\\\midrule
$R_{10}$ & $(a\ d)(b\ c)$\\
$R_{26}$ & $(a)(b\ c)(d)$\\
$R_{27}$ & $(a)(b\ c)(d)$\\
$R_{28}$ & $(a)(b\ c)(d)$\\
$R_{29}$ & $(a\ d)(b\ c)$\\
$R_{43}$ & $(a\ d)(b\ c)$\\
$R_{45}$ & $(a\ b\ d\ c)$\\\bottomrule
\caption{The relevant profile automorphisms, written as a product of their orbits.}
\label{tbl_automorphisms}
\end{longtable}
\end{center}

There are efficiency conditions of two different types: those derived from \emph{ex post} efficiency alone assert that Pareto dominated alternatives have to be assigned probability $0$, whereas those derived from $\sd$-efficiency (but not \emph{ex post} efficiency) assert that at least one of two alternatives has to be assigned probability $0$.

Alternative $b$ is Pareto dominated in the following profiles and must therefore be assigned probability $0$ by any \emph{ex post} efficient SDS (and thereby also by any \sd-efficient SDS):
\begin{displayquote}
$R_{3}$, $R_{4}$, $R_{5}$, $R_{7}$, $R_{8}$, $R_{9}$, $R_{11}$, $R_{12}$, $R_{14}$, $R_{16}$, $R_{17}$, $R_{18}$, $R_{21}$, $R_{22}$, $R_{23}$, $R_{30}$, $R_{32}$, $R_{33}$, $R_{35}$, $R_{40}$, $R_{41}$, $R_{43}$, $R_{44}$, $R_{47}$
\end{displayquote}
We will use the fact that $f(R)(b) = 0$ for all of these profiles without mentioning it explicitly.

Moreover, $\{b,c\}$ is an \sd-inefficient support in the following profiles (\ie any \sd-efficient SDS must assign probability $0$ to at least one of $b$ and $c$):
\begin{displayquote}
$R_{10}$, $R_{15}$, $R_{19}$, $R_{25}$, $R_{26}$, $R_{27}$, $R_{28}$, $R_{29}$, $R_{39}$
\end{displayquote}
To see that this is true, note that the lottery $\nicefrac{1}{2}\,a + \nicefrac{1}{2}\,d$ strictly Pareto dominates the lottery $\nicefrac{1}{2}\,b + \nicefrac{1}{2}\,c$ for each of these profiles.

Using the orbit and efficiency conditions we arrive at the following conclusions:
\begin{itemize}
\item The orbit conditions of $R_{45}$ imply $p_{45,a} = p_{45,b} = p_{45,c} = p_{45,d} = \nicefrac{1}{4}$. 
\item The efficiency conditions for $R_{10}$ state that at least one of $p_{10,b}$ and $p_{10,c}$ is $0$, and since the orbit conditions state that $p_{10,b} = p_{10,c}$, we have $p_{10,b} = p_{10,c} = 0$.
\item In the same fashion, we can show that $p_{i,x} = 0$ for $i\in \{26,27,28,29\}$ and $x\in\{b,c\}$. For $R_{29}$, the orbit condition then additionally implies $p_{29,a} = p_{29,d} = \nicefrac{1}{2}$, and analogously for $R_{10}$.
\item The efficiency conditions for $R_{43}$ state $p_{43,b} = p_{43,c} = 0$, and with the orbit condition $p_{43,a} = p_{43,d}$ we have $p_{43,a} = p_{43,d} = \nicefrac{1}{2}$.
\end{itemize}
In summary, we have now derived the following information about the profiles:
\begin{center}
\begin{tabular}{c|ccccccc}
& $R_{10}$ & $R_{26}$ & $R_{27}$ & $R_{28}$ & $R_{29}$ & $R_{43}$ & $R_{45}$\\\midrule
$a$ & $\nicefrac{1}{2}$ & & & & $\nicefrac{1}{2}$ & $\nicefrac{1}{2}$ & $\nicefrac{1}{4}$\\
$b$ & $0$ & $0$ & $0$ & $0$ & $0$ & $0$ & $\nicefrac{1}{4}$\\
$c$ & $0$ & $0$ & $0$ & $0$ & $0$ & $0$ & $\nicefrac{1}{4}$\\
$d$ & $\nicefrac{1}{2}$ & & &  &  $\nicefrac{1}{2}$ & $\nicefrac{1}{2}$ & $\nicefrac{1}{4}$\\
\end{tabular}
\end{center}

\begin{itemize}
\item Suppose $p_{39,c} = 0$. Then $\eqref{S_29_39}$ implies $p_{39,d} \leq \nicefrac{1}{2}$ and $\eqref{S_39_29}$ then implies $p_{39,b} = 0$. Since the efficiency condition for $R_{39}$ states that $p_{39,b} = 0$ or $p_{39,c} = 0$, we can conclude that, in any case, $p_{39,b} = 0$.
\item Using this, $\eqref{S_39_29}$ now simplifies to $p_{39,a} \leq \nicefrac{1}{2}$.
\item \eqref{S_10_36} simplifies to $p_{36,a} + p_{36,b} \leq \nicefrac{1}{2}$. Using this, \eqref{S_36_10} simplifies to $p_{36,a} = \nicefrac{1}{2}$ and $p_{36,b} = 0$.
\item \eqref{S_36_39} simplifies to $p_{39,a} \geq \nicefrac{1}{2}$. Using this, \eqref{S_39_36} simplifies to $p_{39,a} = \nicefrac{1}{2}$.
\item \eqref{S_12_10} simplifies to $p_{12,a} + p_{12,d} \geq 1$, which implies $p_{12,c} = 0$.
\item \eqref{S_10_12} then simplifies to $p_{12,a} \geq \nicefrac{1}{2}$.
\item $\eqref{S_12_44}$ simplifies to $p_{44,a} \leq p_{12,a}$. Using this, $\eqref{S_44_12}$ simplifies to $p_{44,a} = p_{12,a}$ and $p_{44,c} = 0$.
\item $\eqref{S_9_35}$ simplifies to $p_{35,a} \leq p_{9,a}$, and then $\eqref{S_35_9}$ simplifies to $p_{9,a} = p_{35,a}$.
\item $\eqref{S_9_18}$ states that $p_{9,a} + p_{9,d} \leq p_{18,a} + p_{18,d}$, and then $\eqref{S_9_18}$ simplifies to $p_{18,c} = p_{9,c}$.
\end{itemize}
To summarize:
\begin{center}
\begin{tabular}{c|ccccccccccccc}
& $R_{9}$ & $R_{10}$ & $R_{12}$ & $R_{18}$ & $R_{26}$ & $R_{27}$ & $R_{28}$ & $R_{29}$ & $R_{36}$ & $R_{39}$ & $R_{43}$ & $R_{44}$ & $R_{45}$\\\midrule
$a$ & $p_{35,a}$ & $\nicefrac{1}{2}$ & ${\geq}\,\nicefrac{1}{2}$ &  &  &  &  & $\nicefrac{1}{2}$ & $\nicefrac{1}{2}$ & $\nicefrac{1}{2}$ & $\nicefrac{1}{2}$ & $p_{12,a}$ & $\nicefrac{1}{4}$\\
$b$ & $0$ & $0$ & $0$ & $0$ & $0$ & $0$ & $0$ & $0$ & $0$ & $0$ & $0$ & $0$ & $\nicefrac{1}{4}$\\
$c$ &  & $0$ & $0$ & $p_{9,c}$ & $0$ & $0$ & $0$ & $0$ &  &  & $0$ & $0$ & $\nicefrac{1}{4}$\\
$d$ &  & $\nicefrac{1}{2}$ & ${\leq}\,\nicefrac{1}{2}$ &  &  &  &  &  & $\nicefrac{1}{2}$ &  & $\nicefrac{1}{2}$ & $1-p_{12,a}$ & $\nicefrac{1}{4}$\\
\end{tabular}
\end{center}

\begin{itemize}
\item \eqref{S_5_10} implies $p_{5,d} \geq \nicefrac{1}{2}$.
\item \eqref{S_5_17} implies $p_{5,d} \leq p_{17,d}$, and \eqref{S_17_7} simplifies to $p_{17,d} \leq p_{7,d}$. Combined with $p_{5,d} \geq \nicefrac{1}{2}$ from above, we have $p_{7,d} \geq \nicefrac{1}{2}$. Using this, \eqref{S_7_43} implies $p_{7,a} = \nicefrac{1}{2}$ and $p_{7,c} = 0$, and therefore $p_{7,d}=\nicefrac{1}{2}$.
\item \eqref{S_5_7} now simplifies to $p_{5,d} \leq \nicefrac{1}{2}$, and $p_{5,d} \geq \nicefrac{1}{2}$ was already shown, so we have $p_{5,d} = \nicefrac{1}{2}$.
\item \eqref{S_5_10} now simplifies to $p_{5,c} = 0$, and it is then clear that $p_{5,a} = \nicefrac{1}{2}$.
\item Suppose $p_{15,b} = 0$. Then \eqref{S_10_15} simplifies to $p_{15,a} + p_{15,c} \leq \nicefrac{1}{2}$ and, using that, \eqref{S_15_10} implies $p_{15,c} = 0$. Since the efficiency conditions for $R_{15}$ tell us that $p_{15,b} = 0$ or $p_{15,c} = 0$, we can conclude $p_{15,c} = 0$.
\item \eqref{S_15_5} then implies $p_{15,a} \geq \nicefrac{1}{2}$ and \eqref{S_15_7} implies $p_{15,a} \leq \nicefrac{1}{2}$. We can conclude that $p_{15,a} = \nicefrac{1}{2}$.
\item \eqref{S_15_5} now simplifies to $p_{15,d} = \nicefrac{1}{2}$ and $p_{15,b} = 0$.
\item \eqref{S_27_13} simplifies to $p_{13,a} + p_{13,b} \leq p_{27,a}$. Using that, \eqref{S_13_27} simplifies to $p_{13,b} = p_{13,c} = 0$ and $p_{27,a} = p_{13,a}$.
\item \eqref{S_15_13} now implies $p_{13,a} \geq \nicefrac{1}{2}$ and \eqref{S_13_15} simplifies to $p_{13,a} \leq \nicefrac{1}{2}$, so that we can conclude $p_{13,a} = p_{13,d} = p_{27,a} = p_{27,d} = \nicefrac{1}{2}$.
\end{itemize}
We summarize what we have learned so far:
\begin{center}
\begin{varwidth}{\textwidth}
\begin{tabular}{c|cccccccccccc}
& $R_{5}$ & $R_{7}$ & $R_{9}$ & $R_{10}$ & $R_{12}$ & $R_{13}$ & $R_{15}$ & $R_{18}$ & $R_{26}$ & $R_{27}$ & $R_{28}$ & $R_{29}$\\\midrule
$a$ & $\nicefrac{1}{2}$ & $\nicefrac{1}{2}$ & $p_{35,a}$ & $\nicefrac{1}{2}$ & ${\geq}\,\nicefrac{1}{2}$ & $\nicefrac{1}{2}$ & $\nicefrac{1}{2}$ &  &  & $\nicefrac{1}{2}$ &  & $\nicefrac{1}{2}$\\
$b$ & $0$ & $0$ & $0$ & $0$ & $0$ & $0$ & $0$ & $0$ & $0$ & $0$ & $0$ & $0$\\
$c$ & $0$ & $0$ &  & $0$ & $0$ & $0$ & $0$ & $p_{9,c}$ & $0$ & $0$ & $0$ & $0$\\
$d$ & $\nicefrac{1}{2}$ & $\nicefrac{1}{2}$ &  & $\nicefrac{1}{2}$ & ${\leq}\,\nicefrac{1}{2}$ & $\nicefrac{1}{2}$ & $\nicefrac{1}{2}$ &  &  & $\nicefrac{1}{2}$ &  & $\nicefrac{1}{2}$\\
\end{tabular}
\\[1em]
\begin{tabular}{c|ccccc}
& $R_{36}$ & $R_{39}$ & $R_{43}$ & $R_{44}$ & $R_{45}$\\\midrule
$a$ & $\nicefrac{1}{2}$ & $\nicefrac{1}{2}$ & $\nicefrac{1}{2}$ & $p_{12,a}$ & $\nicefrac{1}{4}$\\
$b$ & $0$ & $0$ & $0$ & $0$ & $\nicefrac{1}{4}$\\
$c$ &  &  & $0$ & $0$ & $\nicefrac{1}{4}$\\
$d$ &  &  & $\nicefrac{1}{2}$ & $1-p_{12,a}$ & $\nicefrac{1}{4}$\\
\end{tabular}
\end{varwidth}
\end{center}

\begin{itemize}
\item We will now determine the probabilities for $R_{19}$. The efficiency condition tells us that $p_{19,b}=0$ or $p_{19,c}=0$.
\begin{itemize}
\item Suppose $p_{19,b} = 0$. Then \eqref{S_10_19} simplifies to $p_{19,a} + p_{19,c} \leq \nicefrac{1}{2}$ and \eqref{S_19_10} simplifies to $p_{19,a} + p_{19,c} = \nicefrac{1}{2}$. We can therefore conclude that $p_{19,d} = \nicefrac{1}{2}$. Using this, \eqref{S_27_19} then simplifies to $p_{19,a} = \nicefrac{1}{2}$ and $p_{19,c} = 0$ and therefore $p_{19,d} = \nicefrac{1}{2}$.
\item Suppose $p_{19,c} = 0$. Then \eqref{S_19_10} simplifies to $p_{19,a} \geq \nicefrac{1}{2}$ and \eqref{S_19_27} simplifies to $p_{19,d} \geq \nicefrac{1}{2}$. This clearly implies $p_{19,a} = p_{19,d} = \nicefrac{1}{2}$ and $p_{19,b} = 0$.
\end{itemize}
\item Using this, \eqref{S_19_1} simplifies to $p_{1,a} + p_{1,b} \leq \nicefrac{1}{2}$, and with that, \eqref{S_1_19} simplifies to $p_{1,a} = \nicefrac{1}{2}$ and $p_{1,b} = 0$.
\item \eqref{S_33_5} simplifies to $p_{33,a} \geq \nicefrac{1}{2}$. Moreover, \eqref{S_33_22} simplifies to $p_{22,c} + p_{22,d} \leq p_{33,c}+p_{33,d}$, \ie $p_{33,a} \leq p_{22,a}$. We therefore have $p_{22,a} \geq \nicefrac{1}{2}$. Using this, \eqref{S_22_29} simplifies to $p_{22,a} = p_{22,d} = \nicefrac{1}{2}$ and therefore also $p_{22,c} = 0$.
\item \eqref{S_32_28} implies $p_{28,a} \leq p_{32,d}$. Then \eqref{S_28_32} implies $p_{32,d} = p_{28,a}$. Moreover, \eqref{S_22_32} simplifies to $p_{32,a} \leq 1$. Using these two facts, \eqref{S_32_22} implies $p_{32,d} = \nicefrac{1}{2}$ and therefore also $p_{28,a} = p_{28,d} = \nicefrac{1}{2}$.
\item \eqref{S_28_39} now simplifies to $p_{39,c} = 0$, and since we have already determined $p_{39,a} = \nicefrac{1}{2}$ and $p_{39,b} = 0$, we can conclude $p_{39,d} = \nicefrac{1}{2}$.
\item \eqref{S_1_2} states that $p_{2,c} + p_{2,d} \leq p_{1,c} + p_{2,d}$. Using this, \eqref{S_2_1} simplifies to $p_{2,a} = p_{2,c} + p_{2,d} = \nicefrac{1}{2}$ and therefore also $p_{2,b} = 0$. Using this, \eqref{S_39_2} simplifies to $p_{2,c} = 0$ and $p_{2,d} = \nicefrac{1}{2}$.
\item We will now determine $R_{42}$:
\begin{itemize}
\item \eqref{S_17_5} simplifies to $p_{17,a} + p_{17,c} \geq \nicefrac{1}{2}$ and \eqref{S_5_17} simplifies to $p_{17,a} + p_{17,c} \leq \nicefrac{1}{2}$, so we can conclude $p_{17,d} = \nicefrac{1}{2}$.
\item \eqref{S_6_42} states that $p_{42,a} + p_{42,c} \leq p_{6,a} + p_{6,c}$ and \eqref{S_6_19} implies $p_{6,a} + p_{6,c} \leq \nicefrac{1}{2}$. We can therefore conclude that $p_{42,a} + p_{42,c} \leq \nicefrac{1}{2}$.
\item \eqref{S_17_11} states that $p_{11,a} + p_{11,c} \leq p_{17,a} + p_{17,c}$. Since $p_{11,b} = p_{17,b} = 0$, this is equivalent to $p_{11,d} \geq p_{17,d} = \nicefrac{1}{2} \geq p_{42,a} + p_{42,c}$. With this, \eqref{S_42_11} implies $p_{42,c} \geq p_{11,d} \geq \nicefrac{1}{2}$.
\item \eqref{S_17_3} simplifies to $p_{3,a} + p_{3,c} \leq p_{17,a} + p_{17,c}$; \ie $p_{3,d} \geq p_{17,d} = \nicefrac{1}{2}$.
\item Finally, using $p_{42,c} \geq \nicefrac{1}{2}$ and $p_{3,d} \geq \nicefrac{1}{2}$, \eqref{S_42_3} simplifies to $p_{42,c} \geq\nicefrac{1}{2}$ and $p_{42,d} \geq \nicefrac{1}{2}$ and therefore $p_{42,a} = p_{42,b} = 0$ and $p_{42,c} = p_{42,d} = \nicefrac{1}{2}$.
\end{itemize}
\item Using these values for $R_{42}$, the two conditions \eqref{S_37_42_1} and \eqref{S_37_42_2} now simplify to $p_{37,a} = \nicefrac{1}{2}$ or $p_{37,a}+p_{37,b} > \nicefrac{1}{2}$, and $p_{37,c} = \nicefrac{1}{2}$ or $p_{37,c}+p_{37,d} > \nicefrac{1}{2}$. Together, these obviously imply $p_{37,a} = p_{37,c} = \nicefrac{1}{2}$ and $p_{37,b} = p_{37,d} = 0$.
\item Similarly, $R_{24}$ simplifies to $p_{24,a} + p_{24,b} \leq 0$ and therefore $p_{24,a} = p_{24,b} = 0$.
\end{itemize}

\begin{center}
\begin{varwidth}{\textwidth}
\begin{tabular}{c|ccccccccccccc}
& $R_{1}$ & $R_{2}$ & $R_{5}$ & $R_{7}$ & $R_{9}$ & $R_{10}$ & $R_{12}$ & $R_{13}$ & $R_{15}$ & $R_{18}$ & $R_{19}$ & $R_{22}$ & $R_{24}$\\\midrule
$a$ & $\nicefrac{1}{2}$ & $\nicefrac{1}{2}$ & $\nicefrac{1}{2}$ & $\nicefrac{1}{2}$ & $p_{35,a}$ & $\nicefrac{1}{2}$ & ${\geq}\,\nicefrac{1}{2}$ & $\nicefrac{1}{2}$ & $\nicefrac{1}{2}$ &  & $\nicefrac{1}{2}$ & $\nicefrac{1}{2}$ & $0$\\
$b$ & $0$ & $0$ & $0$ & $0$ & $0$ & $0$ & $0$ & $0$ & $0$ & $0$ & $0$ & $0$ & $0$\\
$c$ &  & $0$ & $0$ & $0$ &  & $0$ & $0$ & $0$ & $0$ & $p_{9,c}$ & $0$ & $0$ & \\
$d$ &  & $\nicefrac{1}{2}$ & $\nicefrac{1}{2}$ & $\nicefrac{1}{2}$ &  & $\nicefrac{1}{2}$ & ${\leq}\,\nicefrac{1}{2}$ & $\nicefrac{1}{2}$ & $\nicefrac{1}{2}$ &  & $\nicefrac{1}{2}$ & $\nicefrac{1}{2}$ & \\
\end{tabular}
\\[1em]
\begin{tabular}{c|ccccccccccc}
& $R_{26}$ & $R_{27}$ & $R_{28}$ & $R_{29}$ & $R_{36}$ & $R_{37}$ & $R_{39}$ & $R_{42}$ & $R_{43}$ & $R_{44}$ & $R_{45}$\\\midrule
$a$ &  & $\nicefrac{1}{2}$ & $\nicefrac{1}{2}$ & $\nicefrac{1}{2}$ & $\nicefrac{1}{2}$ & $\nicefrac{1}{2}$ & $\nicefrac{1}{2}$ & $0$ & $\nicefrac{1}{2}$ & $p_{12,a}$ & $\nicefrac{1}{4}$\\
$b$ & $0$ & $0$ & $0$ & $0$ & $0$ & $0$ & $0$ & $0$ & $0$ & $0$ & $\nicefrac{1}{4}$\\
$c$ & $0$ & $0$ & $0$ & $0$ &  & $\nicefrac{1}{2}$ & $0$ & $\nicefrac{1}{2}$ & $0$ & $0$ & $\nicefrac{1}{4}$\\
$d$ &  & $\nicefrac{1}{2}$ & $\nicefrac{1}{2}$ & $\nicefrac{1}{2}$ &  & $0$ & $\nicefrac{1}{2}$ & $\nicefrac{1}{2}$ & $\nicefrac{1}{2}$ & $1-p_{12,a}$ & $\nicefrac{1}{4}$\\
\end{tabular}
\end{varwidth}
\end{center}

\begin{itemize}
\item \eqref{S_24_34} implies $p_{34,b} \leq p_{24,c}$ and \eqref{S_34_24} implies $p_{24,c} \leq p_{34,b}$; we therefore have $p_{34,b} = p_{24,c}$. Using this, \eqref{S_34_24} simplifies to $p_{34,c} = 0$ and \eqref{S_24_34} simplifies to $p_{34,d} = 0$.
\item \eqref{S_14_34} now simplifies to $p_{14,a} + p_{14,c} \geq 1$, so we have $p_{14,b} = p_{14,d} = 0$.
\item \eqref{S_46_37} simplifies to $p_{46,a} = p_{46,c} = 0$.
\item \eqref{S_46_20} now simplifies to $p_{20,a} + p_{20,c} \leq 0$, so we have $p_{20,a} = p_{20,c} = 0$.
\item \eqref{S_20_21} now simplifies to $p_{21,b} = p_{21,c} = 0$.
\item \eqref{S_12_16} simplifies to $p_{16,a} + p_{16,c} \leq p_{12,a}$.
\item We now determine the probabilities for $p_{16,c}$:
\begin{itemize}
\item \eqref{S_44_40} simplifies to $p_{12,a} \leq p_{40,a}$. Moreover, \eqref{S_9_40} simplifies to $p_{40,a} \leq p_{9,a}$. Combined with $p_{16,a} + p_{16,c} \leq p_{12,a}$, this implies $p_{16,a} + p_{16,c} \leq p_{9,a}$.
\item \eqref{S_14_16} implies $p_{16,a} \geq p_{14,a}$.
\item Combining the last two facts, we obtain $p_{16,c} \leq p_{9,a} - p_{14,a}$. Moreover, \eqref{S_14_9} implies $p_{9,a} - p_{14,a} \leq 0$. Combining this, we have $p_{16,c} = 0$.
\end{itemize}
\item Therefore, the fact $p_{16,a} + p_{16,c} \leq p_{12,a}$, which we have shown before, now simplifies to $p_{16,a} \leq p_{12,a}$.
\item Since \eqref{S_14_16} simplifies to $p_{14,a} \leq p_{16,a}$, we then have $p_{14,a} \leq p_{12,a}$.
\end{itemize}

\begin{center}
\begin{varwidth}{\textwidth}
\begin{tabular}{c|ccccccccccccc}
& $R_{1}$ & $R_{2}$ & $R_{5}$ & $R_{7}$ & $R_{9}$ & $R_{10}$ & $R_{12}$ & $R_{13}$ & $R_{14}$ & $R_{15}$ & $R_{16}$ & $R_{18}$ & $R_{19}$\\\midrule
$a$ & $\nicefrac{1}{2}$ & $\nicefrac{1}{2}$ & $\nicefrac{1}{2}$ & $\nicefrac{1}{2}$ & $p_{35,a}$ & $\nicefrac{1}{2}$ & ${\geq}\,\nicefrac{1}{2}$ & $\nicefrac{1}{2}$ & ${\leq}\,p_{12,a}$ & $\nicefrac{1}{2}$ &  &  & $\nicefrac{1}{2}$\\
$b$ & $0$ & $0$ & $0$ & $0$ & $0$ & $0$ & $0$ & $0$ & $0$ & $0$ & $0$ & $0$ & $0$\\
$c$ &  & $0$ & $0$ & $0$ &  & $0$ & $0$ & $0$ &  & $0$ & $0$ & $p_{9,c}$ & $0$\\
$d$ &  & $\nicefrac{1}{2}$ & $\nicefrac{1}{2}$ & $\nicefrac{1}{2}$ &  & $\nicefrac{1}{2}$ & ${\leq}\,\nicefrac{1}{2}$ & $\nicefrac{1}{2}$ & $0$ & $\nicefrac{1}{2}$ &  &  & $\nicefrac{1}{2}$\\
\end{tabular}
\\[1em]
\begin{tabular}{c|cccccccccccc}
& $R_{20}$ & $R_{21}$ & $R_{22}$ & $R_{24}$ & $R_{26}$ & $R_{27}$ & $R_{28}$ & $R_{29}$ & $R_{34}$ & $R_{36}$ & $R_{37}$ & $R_{39}$\\\midrule
$a$ & $0$ &  & $\nicefrac{1}{2}$ & $0$ &  & $\nicefrac{1}{2}$ & $\nicefrac{1}{2}$ & $\nicefrac{1}{2}$ & $1-p_{24,c}$ & $\nicefrac{1}{2}$ & $\nicefrac{1}{2}$ & $\nicefrac{1}{2}$\\
$b$ &  & $0$ & $0$ & $0$ & $0$ & $0$ & $0$ & $0$ & $p_{24,c}$ & $0$ & $0$ & $0$\\
$c$ & $0$ & $0$ & $0$ &  & $0$ & $0$ & $0$ & $0$ & $0$ &  & $\nicefrac{1}{2}$ & $0$\\
$d$ &  &  & $\nicefrac{1}{2}$ &  &  & $\nicefrac{1}{2}$ & $\nicefrac{1}{2}$ & $\nicefrac{1}{2}$ & $0$ &  & $0$ & $\nicefrac{1}{2}$\\
\end{tabular}
\\[1em]
\begin{tabular}{c|ccccc}
& $R_{42}$ & $R_{43}$ & $R_{44}$ & $R_{45}$ & $R_{46}$\\\midrule
$a$ & $0$ & $\nicefrac{1}{2}$ & $p_{12,a}$ & $\nicefrac{1}{4}$ & $0$\\
$b$ & $0$ & $0$ & $0$ & $\nicefrac{1}{4}$ & \\
$c$ & $\nicefrac{1}{2}$ & $0$ & $0$ & $\nicefrac{1}{4}$ & $0$\\
$d$ & $\nicefrac{1}{2}$ & $\nicefrac{1}{2}$ & $1-p_{12,a}$ & $\nicefrac{1}{4}$ & \\
\end{tabular}
\end{varwidth}
\end{center}

\begin{itemize}
\item We now show that $p_{12,a} = p_{9,a} = p_{35,a}$:
\begin{itemize}
\item \eqref{S_14_9} implies $p_{9,a} \leq p_{14,a}$. Since $p_{14,a} \leq p_{12,a}$, we have $p_{9,a} \leq p_{12,a}$.
\item \eqref{S_44_40} simplifies to $p_{12,a} \leq p_{40,a}$. Moreover, \eqref{S_9_40} simplifies to $p_{40,a} \leq p_{9,a}$; therefore, we have $p_{12,a} \leq p_{9,a}$.
\item Combining these two inequalities yields $p_{12,a} = p_{9,a}$.
\end{itemize}
\item Recall that $p_{14,a} \leq p_{12,a} = p_{9,a}$. Then \eqref{S_14_9} simplifies to $p_{9,a} = p_{14,a}$ and $p_{9,d} = 0$.
\item \eqref{S_23_19} simplifies to $p_{23,a} + p_{23,d} \geq 1$ and therefore $p_{23,b} = p_{23,c} = 0$.
\item \eqref{S_35_21} simplifies to $p_{21,a} \leq p_{35,a} + p_{35,c}$. Then \eqref{S_21_35} simplifies to $p_{35,c} = 0$ and $p_{35,a} = p_{21,a}$.
\item Next, we derive the probabilities for $R_{18}$:
\begin{itemize}
\item \eqref{S_23_12} simplifies to $p_{21,a} \leq p_{23,a}$.
\item \eqref{S_23_18} simplifies to $p_{18,c} + p_{18,d} \leq 1 - p_{23,a}$. Since $p_{18,c} = p_{9,c} = 1 - p_{9,a} = 1 - p_{35,a} = 1 - p_{21,a}$, this is equivalent to $p_{18,d} \leq p_{21,a} - p_{23,a}$.
Recall that $p_{9,b} = p_{9,c} = 0$, \ie $p_{18,c} = p_{9,c} = 1 - p_{9,a} = 1 - p_{35,a} = 1 - p_{21,a}$. Substituting this in the inequality we have just derived and rearranging yields 
$p_{18,d} \leq p_{21,a} - p_{23,a}$.
\item Since $p_{21,a} \leq p_{23,a}$, the right-hand side of the above inequality is $0$ and therefore $p_{18,d} = 0$.
\end{itemize}
Now we can derive the probabilities for $R_4$:
\begin{itemize}
\item \eqref{S_47_30} simplifies to $p_{30,a} \leq p_{47,a}$.
\item \eqref{S_4_47} simplifies to $p_{47,a} + p_{47,d} \leq p_{4,a} + p_{4,d}$, \ie $p_{4,c} \leq p_{47,c}$.
\item Adding these two inequalities, we obtain $p_{4,c} + p_{30,a} \leq 1 - p_{47,d}$.
\item \eqref{S_30_21} simplifies to $p_{21,a} \leq p_{30,a}$, and with the previous inequality, we obtain $p_{4,c} + p_{21,a} \leq 1 - p_{47,d} \leq 1$. Substituting $p_{21,a} = p_{14,a}$ yields $p_{4,c} + p_{14,a} \leq 1$.
\item \eqref{S_4_18} now simplifies to $p_{4,d} = 0$ and $p_{4,c} = p_{21,d}$.
\end{itemize}
\item \eqref{S_8_26} implies $p_{26,a} \leq p_{8,d}$. Using this, \eqref{S_26_8} simplifies to $p_{26,a} = p_{8,d}$. Using this, we look at \eqref{S_8_26} again and find that it now simplifies to $p_{8,a} + p_{8,d} = 1$, \ie $p_{8,c} = p_{8,b} = 0$ and $p_{26,a} = 1 - p_{8,a}$.
\end{itemize}

\begin{center}
\begin{varwidth}{\textwidth}
\begin{tabular}{c|ccccccccccc}
& $R_{1}$ & $R_{2}$ & $R_{4}$ & $R_{5}$ & $R_{7}$ & $R_{8}$ & $R_{9}$ & $R_{10}$ & $R_{12}$ & $R_{13}$ & $R_{14}$\\\midrule
$a$ & $\nicefrac{1}{2}$ & $\nicefrac{1}{2}$ & $p_{21,a}$ & $\nicefrac{1}{2}$ & $\nicefrac{1}{2}$ &  & $p_{21,a}$ & $\nicefrac{1}{2}$ & $p_{21,a}$ & $\nicefrac{1}{2}$ & $p_{21,a}$\\
$b$ & $0$ & $0$ & $0$ & $0$ & $0$ & $0$ & $0$ & $0$ & $0$ & $0$ & $0$\\
$c$ &  & $0$ & $1-p_{21,a}$ & $0$ & $0$ & $0$ & $1-p_{21,a}$ & $0$ & $0$ & $0$ & $1-p_{21,a}$\\
$d$ &  & $\nicefrac{1}{2}$ & $0$ & $\nicefrac{1}{2}$ & $\nicefrac{1}{2}$ &  & $0$ & $\nicefrac{1}{2}$ & $1-p_{21,a}$ & $\nicefrac{1}{2}$ & $0$\\
\end{tabular}
\\[1em]
\begin{tabular}{c|cccccccccccc}
& $R_{15}$ & $R_{16}$ & $R_{18}$ & $R_{19}$ & $R_{20}$ & $R_{21}$ & $R_{22}$ & $R_{23}$ & $R_{24}$ & $R_{26}$ & $R_{27}$ & $R_{28}$\\\midrule
$a$ & $\nicefrac{1}{2}$ &  & $p_{21,a}$ & $\nicefrac{1}{2}$ & $0$ &  & $\nicefrac{1}{2}$ &  & $0$ & $1-p_{8,a}$ & $\nicefrac{1}{2}$ & $\nicefrac{1}{2}$\\
$b$ & $0$ & $0$ & $0$ & $0$ &  & $0$ & $0$ & $0$ & $0$ & $0$ & $0$ & $0$\\
$c$ & $0$ & $0$ & $1-p_{21,a}$ & $0$ & $0$ & $0$ & $0$ & $0$ &  & $0$ & $0$ & $0$\\
$d$ & $\nicefrac{1}{2}$ &  & $0$ & $\nicefrac{1}{2}$ &  &  & $\nicefrac{1}{2}$ &  &  & $p_{8,a}$ & $\nicefrac{1}{2}$ & $\nicefrac{1}{2}$\\
\end{tabular}
\\[1em]
\begin{tabular}{c|ccccccccccc}
& $R_{29}$ & $R_{34}$ & $R_{35}$ & $R_{36}$ & $R_{37}$ & $R_{39}$ & $R_{42}$ & $R_{43}$ & $R_{44}$ & $R_{45}$ & $R_{46}$\\\midrule
$a$ & $\nicefrac{1}{2}$ & $1-p_{24,c}$ & $p_{21,a}$ & $\nicefrac{1}{2}$ & $\nicefrac{1}{2}$ & $\nicefrac{1}{2}$ & $0$ & $\nicefrac{1}{2}$ & $p_{12,a}$ & $\nicefrac{1}{4}$ & $0$\\
$b$ & $0$ & $p_{24,c}$ & $0$ & $0$ & $0$ & $0$ & $0$ & $0$ & $0$ & $\nicefrac{1}{4}$ & \\
$c$ & $0$ & $0$ & $0$ &  & $\nicefrac{1}{2}$ & $0$ & $\nicefrac{1}{2}$ & $0$ & $0$ & $\nicefrac{1}{4}$ & $0$\\
$d$ & $\nicefrac{1}{2}$ & $0$ & $1-p_{21,a}$ &  & $0$ & $\nicefrac{1}{2}$ & $\nicefrac{1}{2}$ & $\nicefrac{1}{2}$ & $1-p_{12,a}$ & $\nicefrac{1}{4}$ & \\
\end{tabular}
\end{varwidth}
\end{center}

\begin{itemize}
\item \eqref{S_4_47} simplifies to $p_{21,d} \leq p_{47,c}$.
\item \eqref{S_47_30} simplifies to $p_{30,a} \leq p_{47,a}$. With this and the previous inequality, \eqref{S_30_21} simplifies to $p_{30,b} = p_{30,c} = 0$ and $p_{30,a} = p_{47,a}$. 
\item The last big and crucial step is to show that $p_{31,c} \geq \nicefrac{1}{2}$:
\begin{itemize}
\item The efficiency conditions for $R_{25}$ tell us that $p_{25,b} = 0$ or $p_{25,c} = 0$. If $p_{25,c} = 0$, then \eqref{S_25_36} immediately implies $p_{25,a} \geq \nicefrac{1}{2}$. If, on the other hand, $p_{25,b} = 0$, then \eqref{S_36_25} implies $p_{25,a} + p_{25,c} \leq p_{36,c} + \nicefrac{1}{2}$, with which \eqref{S_25_36} then also implies $p_{25,a} \geq \nicefrac{1}{2}$.
\item Using $p_{25,a} \geq \nicefrac{1}{2}$, the condition \eqref{S_25_26} implies $p_{26,a} \geq \nicefrac{1}{2}$, and therefore also $\nicefrac{1}{2} \leq p_{26,a} + p_{47,d} = 1 - p_{8,a} + p_{47,d}$.
\item Now observe that \eqref{S_4_8} simplifies to $p_{21,a} \leq p_{8,a}$, which is equivalent to $1 - p_{8,a} \leq p_{21,d}$. Combined with $p_{21,d} \leq p_{47,c}$, which we have shown before, we now have $\nicefrac{1}{2} \leq p_{47,c} + p_{47,d}$.
\item \eqref{S_30_41} implies $p_{41,a} + p_{41,c} \leq p_{47,a}$, which is equivalent to $p_{47,c} + p_{47,d} \leq p_{41,d}$.
\item \eqref{S_41_31} simplifies to $p_{31,a} + p_{31,b} + p_{31,d} \leq p_{41,a} + p_{41,c}$, which is equivalent to $p_{41,d} \leq p_{31,c}$.
\item Combining this chain of inequalities, we finally have $p_{31,c} \geq \nicefrac{1}{2}$.
\end{itemize}
\item \eqref{S_2_38} simplifies to $p_{38,a} + p_{38,c} \leq \nicefrac{1}{2}$, \ie $p_{38,b} + p_{38,d} \geq \nicefrac{1}{2}$. Using this and $p_{31,c} \geq \nicefrac{1}{2}$, the condition \eqref{S_31_38} simplifies to $p_{38,b} + p_{38,d} = p_{31,b} + p_{31,d}$. This means that $p_{31,b} + p_{31,d} \geq \nicefrac{1}{2}$, and since $p_{31,c} \geq \nicefrac{1}{2}$, we can conclude $p_{31,b} + p_{31,d} = p_{31,c} = \nicefrac{1}{2}$ and $p_{31,a} = 0$.
\end{itemize}
It is now easy to see that each of the three cases in \eqref{S_45_31} is a contradiction. We have thus shown that the conditions are inconsistent, and therefore, there is no anonymous and neutral SDS for four agents and alternatives that satisfies both strategyproofness and efficiency.
\hfill\qedsymbol

\subsection{Strategyproofness Conditions}

\Cref{tbl:sp} lists the strategyproofness conditions that were used in the impossibility proof. As explained in \Cref{sec:domains}, all manipulations are either $1$-manipulations or $2$-manipulations, i.e., a manipulator breaks or introduces a tie between two alternatives or swaps two alternatives.
They are a subset of the conditions derived by the \textit{derive\_\allowbreak strategyproofness\_\allowbreak conditions} command with a distance threshold of 2, \ie the required manipulations all have a size ${\leq}\,2$.
The first number in the name of the condition indicates the original profile and the second one is the manipulated profile (possibly with a permutation applied to the alternatives).

\newpage
\newcommand{\myvee}{\hskip1.5mm {\vee}\hskip1.5mm}
\newcommand{\mywedge}{\hskip1.5mm {\wedge}\hskip1.5mm}
\newcommand{\mybreak}{\\ &\hspace*{-2mm} }
\newcommand{\mypagebreak}{\displaybreak[0]\\[0.5em]}
\begingroup
\addtolength{\jot}{-0.4em}
\begin{fleqn}
\noindent \rule{\textwidth}{\heavyrulewidth}\\[-2.3em]
\begin{gather}
\begin{align*}\MoveEqLeft p_{2,d} + p_{2,c} \leq p_{1,d} + p_{1,c} \end{align*} \label{S_1_2} \tag{$S_{1,2}$}\mypagebreak
\begin{align*}\MoveEqLeft p_{19,a} < p_{1,a} \myvee p_{19,a} + p_{19,b} < p_{1,a} + p_{1,b} \myvee ( p_{19,a} = p_{1,a} \mywedge p_{19,a} + p_{19,b} = p_{1,a} + p_{1,b} ) \end{align*} \label{S_1_19} \tag{$S_{1,19}$}\mypagebreak
\begin{align*}\MoveEqLeft p_{1,d} + p_{1,c} < p_{2,d} + p_{2,c} \myvee p_{1,d} + p_{1,c} + p_{1,a} < p_{2,d} + p_{2,c} + p_{2,a} \mybreak\myvee ( p_{1,d} + p_{1,c} = p_{2,d} + p_{2,c} \mywedge p_{1,d} + p_{1,c} + p_{1,a} = p_{2,d} + p_{2,c} + p_{2,a} ) \end{align*} \label{S_2_1} \tag{$S_{2,1}$}\mypagebreak
\begin{align*}\MoveEqLeft p_{38,c} + p_{38,a} \leq p_{2,c} + p_{2,a} \end{align*} \label{S_2_38} \tag{$S_{2,38}$}\mypagebreak
\begin{align*}\MoveEqLeft p_{8,c} < p_{4,d} \myvee p_{8,c} + p_{8,d} < p_{4,d} + p_{4,c} \myvee ( p_{8,c} = p_{4,d} \mywedge p_{8,c} + p_{8,d} = p_{4,d} + p_{4,c} ) \end{align*} \label{S_4_8} \tag{$S_{4,8}$}\mypagebreak
\begin{align*}\MoveEqLeft p_{18,c} < p_{4,c} \myvee p_{18,c} + p_{18,b} + p_{18,a} < p_{4,c} + p_{4,b} + p_{4,a} \mybreak\myvee ( p_{18,c} = p_{4,c} \mywedge p_{18,c} + p_{18,b} + p_{18,a} = p_{4,c} + p_{4,b} + p_{4,a} ) \end{align*} \label{S_4_18} \tag{$S_{4,18}$}\mypagebreak
\begin{align*}\MoveEqLeft p_{47,d} + p_{47,a} \leq p_{4,d} + p_{4,a} \end{align*} \label{S_4_47} \tag{$S_{4,47}$}\mypagebreak
\begin{align*}\MoveEqLeft p_{7,c} + p_{7,a} < p_{5,c} + p_{5,a} \myvee p_{7,c} + p_{7,a} + p_{7,d} < p_{5,c} + p_{5,a} + p_{5,d} \mybreak\myvee ( p_{7,c} + p_{7,a} = p_{5,c} + p_{5,a} \mywedge p_{7,c} + p_{7,a} + p_{7,d} = p_{5,c} + p_{5,a} + p_{5,d} ) \end{align*} \label{S_5_7} \tag{$S_{5,7}$}\mypagebreak
\begin{align*}\MoveEqLeft p_{10,a} < p_{5,d} \myvee p_{10,a} + p_{10,c} + p_{10,d} < p_{5,d} + p_{5,b} + p_{5,a} \mybreak\myvee ( p_{10,a} = p_{5,d} \mywedge p_{10,a} + p_{10,c} + p_{10,d} = p_{5,d} + p_{5,b} + p_{5,a} ) \end{align*} \label{S_5_10} \tag{$S_{5,10}$}\mypagebreak
\begin{align*}\MoveEqLeft p_{17,c} + p_{17,a} < p_{5,c} + p_{5,a} \myvee p_{17,c} + p_{17,a} + p_{17,d} < p_{5,c} + p_{5,a} + p_{5,d} \mybreak\myvee ( p_{17,c} + p_{17,a} = p_{5,c} + p_{5,a} \mywedge p_{17,c} + p_{17,a} + p_{17,d} = p_{5,c} + p_{5,a} + p_{5,d} ) \end{align*} \label{S_5_17} \tag{$S_{5,17}$}\mypagebreak
\begin{align*}\MoveEqLeft p_{19,d} < p_{6,d} \myvee p_{19,d} + p_{19,b} < p_{6,d} + p_{6,b} \myvee p_{19,d} + p_{19,b} + p_{19,a} < p_{6,d} + p_{6,b} + p_{6,a} \mybreak\myvee ( p_{19,d} = p_{6,d} \mywedge p_{19,d} + p_{19,b} = p_{6,d} + p_{6,b} \mywedge p_{19,d} + p_{19,b} + p_{19,a} = p_{6,d} + p_{6,b} + p_{6,a} ) \end{align*} \label{S_6_19} \tag{$S_{6,19}$}\mypagebreak
\begin{align*}\MoveEqLeft p_{42,c} + p_{42,a} \leq p_{6,c} + p_{6,a} \end{align*} \label{S_6_42} \tag{$S_{6,42}$}\mypagebreak
\begin{align*}\MoveEqLeft p_{43,d} < p_{7,a} \myvee p_{43,d} + p_{43,b} < p_{7,a} + p_{7,c} \myvee p_{43,d} + p_{43,b} + p_{43,a} < p_{7,a} + p_{7,c} + p_{7,d} \mybreak\myvee ( p_{43,d} = p_{7,a} \mywedge p_{43,d} + p_{43,b} = p_{7,a} + p_{7,c} \mywedge p_{43,d} + p_{43,b} + p_{43,a} = p_{7,a} + p_{7,c} + p_{7,d} ) \end{align*} \label{S_7_43} \tag{$S_{7,43}$}\mypagebreak
\begin{align*}\MoveEqLeft p_{26,a} < p_{8,d} \myvee p_{26,a} + p_{26,b} + p_{26,d} < p_{8,d} + p_{8,b} + p_{8,a} \mybreak\myvee ( p_{26,a} = p_{8,d} \mywedge p_{26,a} + p_{26,b} + p_{26,d} = p_{8,d} + p_{8,b} + p_{8,a} ) \end{align*} \label{S_8_26} \tag{$S_{8,26}$}\mypagebreak
\begin{align*}\MoveEqLeft p_{18,d} + p_{18,a} < p_{9,d} + p_{9,a} \myvee p_{18,d} + p_{18,a} + p_{18,c} < p_{9,d} + p_{9,a} + p_{9,c} \mybreak\myvee ( p_{18,d} + p_{18,a} = p_{9,d} + p_{9,a} \mywedge p_{18,d} + p_{18,a} + p_{18,c} = p_{9,d} + p_{9,a} + p_{9,c} ) \end{align*} \label{S_9_18} \tag{$S_{9,18}$}\mypagebreak
\begin{align*}\MoveEqLeft p_{35,b} + p_{35,a} \leq p_{9,b} + p_{9,a} \end{align*} \label{S_9_35} \tag{$S_{9,35}$}\mypagebreak
\begin{align*}\MoveEqLeft p_{40,b} + p_{40,a} \leq p_{9,b} + p_{9,a} \end{align*} \label{S_9_40} \tag{$S_{9,40}$}\mypagebreak
\begin{align*}\MoveEqLeft p_{12,b} + p_{12,d} < p_{10,c} + p_{10,a} \myvee p_{12,b} + p_{12,d} + p_{12,a} < p_{10,c} + p_{10,a} + p_{10,d} \mybreak\myvee ( p_{12,b} + p_{12,d} = p_{10,c} + p_{10,a} \mywedge p_{12,b} + p_{12,d} + p_{12,a} = p_{10,c} + p_{10,a} + p_{10,d} ) \end{align*} \label{S_10_12} \tag{$S_{10,12}$}\mypagebreak
\begin{align*}\MoveEqLeft p_{15,a} + p_{15,c} < p_{10,d} + p_{10,b} \myvee p_{15,a} + p_{15,c} + p_{15,d} < p_{10,d} + p_{10,b} + p_{10,a} \mybreak\myvee ( p_{15,a} + p_{15,c} = p_{10,d} + p_{10,b} \mywedge p_{15,a} + p_{15,c} + p_{15,d} = p_{10,d} + p_{10,b} + p_{10,a} ) \end{align*} \label{S_10_15} \tag{$S_{10,15}$}\mypagebreak
\begin{align*}\MoveEqLeft p_{19,a} + p_{19,c} < p_{10,d} + p_{10,b} \myvee p_{19,a} + p_{19,c} + p_{19,d} < p_{10,d} + p_{10,b} + p_{10,a} \mybreak\myvee ( p_{19,a} + p_{19,c} = p_{10,d} + p_{10,b} \mywedge p_{19,a} + p_{19,c} + p_{19,d} = p_{10,d} + p_{10,b} + p_{10,a} ) \end{align*} \label{S_10_19} \tag{$S_{10,19}$}\mypagebreak
\begin{align*}\MoveEqLeft p_{36,a} + p_{36,b} \leq p_{10,d} + p_{10,c} \end{align*} \label{S_10_36} \tag{$S_{10,36}$}\mypagebreak
\begin{align*}\MoveEqLeft p_{10,a} + p_{10,c} + p_{10,d} \leq p_{12,d} + p_{12,b} + p_{12,a} \end{align*} \label{S_12_10} \tag{$S_{12,10}$}\mypagebreak
\begin{align*}\MoveEqLeft p_{16,c} + p_{16,a} < p_{12,c} + p_{12,a} \myvee p_{16,c} + p_{16,a} + p_{16,d} < p_{12,c} + p_{12,a} + p_{12,d} \mybreak\myvee ( p_{16,c} + p_{16,a} = p_{12,c} + p_{12,a} \mywedge p_{16,c} + p_{16,a} + p_{16,d} = p_{12,c} + p_{12,a} + p_{12,d} ) \end{align*} \label{S_12_16} \tag{$S_{12,16}$}\mypagebreak
\begin{align*}\MoveEqLeft p_{44,b} + p_{44,a} \leq p_{12,b} + p_{12,a} \end{align*} \label{S_12_44} \tag{$S_{12,44}$}\mypagebreak
\begin{align*}\MoveEqLeft p_{15,d} + p_{15,c} < p_{13,d} + p_{13,b} \myvee p_{15,d} + p_{15,c} + p_{15,a} < p_{13,d} + p_{13,b} + p_{13,a} \mybreak\myvee ( p_{15,d} + p_{15,c} = p_{13,d} + p_{13,b} \mywedge p_{15,d} + p_{15,c} + p_{15,a} = p_{13,d} + p_{13,b} + p_{13,a} ) \end{align*} \label{S_13_15} \tag{$S_{13,15}$}\mypagebreak
\begin{align*}\MoveEqLeft p_{27,a} < p_{13,a} \myvee p_{27,a} + p_{27,c} < p_{13,a} + p_{13,b} \myvee p_{27,a} + p_{27,c} + p_{27,d} < p_{13,a} + p_{13,b} + p_{13,d} \mybreak\myvee ( p_{27,a} = p_{13,a} \mywedge p_{27,a} + p_{27,c} = p_{13,a} + p_{13,b} \mywedge\mybreak\phantom{\myvee (} p_{27,a} + p_{27,c} + p_{27,d} = p_{13,a} + p_{13,b} + p_{13,d} ) \end{align*} \label{S_13_27} \tag{$S_{13,27}$}\mypagebreak
\begin{align*}\MoveEqLeft p_{9,a} < p_{14,a} \myvee p_{9,a} + p_{9,d} < p_{14,a} + p_{14,d} \myvee p_{9,a} + p_{9,d} + p_{9,c} < p_{14,a} + p_{14,d} + p_{14,c} \mybreak\myvee ( p_{9,a} = p_{14,a} \mywedge p_{9,a} + p_{9,d} = p_{14,a} + p_{14,d} \mywedge p_{9,a} + p_{9,d} + p_{9,c} = p_{14,a} + p_{14,d} + p_{14,c} ) \end{align*} \label{S_14_9} \tag{$S_{14,9}$}\mypagebreak
\begin{align*}\MoveEqLeft p_{16,c} < p_{14,d} \myvee p_{16,c} + p_{16,d} < p_{14,d} + p_{14,c} \myvee\mybreak ( p_{16,c} = p_{14,d} \mywedge p_{16,c} + p_{16,d} = p_{14,d} + p_{14,c} ) \end{align*} \label{S_14_16} \tag{$S_{14,16}$}\mypagebreak
\begin{align*}\MoveEqLeft p_{34,d} + p_{34,b} + p_{34,a} \leq p_{14,c} + p_{14,b} + p_{14,a} \end{align*} \label{S_14_34} \tag{$S_{14,34}$}\mypagebreak
\begin{align*}\MoveEqLeft p_{5,d} < p_{15,a} \myvee p_{5,d} + p_{5,b} < p_{15,a} + p_{15,c} \myvee p_{5,d} + p_{5,b} + p_{5,a} < p_{15,a} + p_{15,c} + p_{15,d} \mybreak\myvee ( p_{5,d} = p_{15,a} \mywedge p_{5,d} + p_{5,b} = p_{15,a} + p_{15,c} \mywedge p_{5,d} + p_{5,b} + p_{5,a} = p_{15,a} + p_{15,c} + p_{15,d} ) \end{align*} \label{S_15_5} \tag{$S_{15,5}$}\mypagebreak
\begin{align*}\MoveEqLeft p_{7,d} + p_{7,b} < p_{15,d} + p_{15,b} \myvee p_{7,d} + p_{7,b} + p_{7,a} < p_{15,d} + p_{15,b} + p_{15,a} \mybreak\myvee ( p_{7,d} + p_{7,b} = p_{15,d} + p_{15,b} \mywedge p_{7,d} + p_{7,b} + p_{7,a} = p_{15,d} + p_{15,b} + p_{15,a} ) \end{align*} \label{S_15_7} \tag{$S_{15,7}$}\mypagebreak
\begin{align*}\MoveEqLeft p_{10,d} < p_{15,a} \myvee p_{10,d} + p_{10,b} < p_{15,a} + p_{15,c} \myvee p_{10,d} + p_{10,b} + p_{10,a} < p_{15,a} + p_{15,c} + p_{15,d} \mybreak\myvee ( p_{10,d} = p_{15,a} \mywedge p_{10,d} + p_{10,b} = p_{15,a} + p_{15,c} \mywedge\mybreak\phantom{\myvee (} p_{10,d} + p_{10,b} + p_{10,a} = p_{15,a} + p_{15,c} + p_{15,d} ) \end{align*} \label{S_15_10} \tag{$S_{15,10}$}\mypagebreak
\begin{align*}\MoveEqLeft p_{13,d} + p_{13,b} \leq p_{15,d} + p_{15,c} \end{align*} \label{S_15_13} \tag{$S_{15,13}$}\mypagebreak
\begin{align*}\MoveEqLeft p_{3,c} + p_{3,a} \leq p_{17,c} + p_{17,a} \end{align*} \label{S_17_3} \tag{$S_{17,3}$}\mypagebreak
\begin{align*}\MoveEqLeft p_{5,c} + p_{5,a} \leq p_{17,c} + p_{17,a} \end{align*} \label{S_17_5} \tag{$S_{17,5}$}\mypagebreak
\begin{align*}\MoveEqLeft p_{7,c} + p_{7,a} \leq p_{17,c} + p_{17,a} \end{align*} \label{S_17_7} \tag{$S_{17,7}$}\mypagebreak
\begin{align*}\MoveEqLeft p_{11,c} + p_{11,a} \leq p_{17,c} + p_{17,a} \end{align*} \label{S_17_11} \tag{$S_{17,11}$}\mypagebreak
\begin{align*}\MoveEqLeft p_{9,d} + p_{9,a} \leq p_{18,d} + p_{18,a} \end{align*} \label{S_18_9} \tag{$S_{18,9}$}\mypagebreak
\begin{align*}\MoveEqLeft p_{1,b} + p_{1,a} \leq p_{19,b} + p_{19,a} \end{align*} \label{S_19_1} \tag{$S_{19,1}$}\mypagebreak
\begin{align*}\MoveEqLeft p_{10,b} + p_{10,d} \leq p_{19,c} + p_{19,a} \end{align*} \label{S_19_10} \tag{$S_{19,10}$}\mypagebreak
\begin{align*}\MoveEqLeft p_{27,d} + p_{27,b} \leq p_{19,d} + p_{19,c} \end{align*} \label{S_19_27} \tag{$S_{19,27}$}\mypagebreak
\begin{align*}\MoveEqLeft p_{21,c} < p_{20,a} \myvee p_{21,c} + p_{21,b} < p_{20,a} + p_{20,c} \mybreak\myvee ( p_{21,c} = p_{20,a} \mywedge p_{21,c} + p_{21,b} = p_{20,a} + p_{20,c} ) \end{align*} \label{S_20_21} \tag{$S_{20,21}$}\mypagebreak
\begin{align*}\MoveEqLeft p_{35,c} < p_{21,c} \myvee p_{35,c} + p_{35,b} + p_{35,a} < p_{21,c} + p_{21,b} + p_{21,a} \mybreak\myvee ( p_{35,c} = p_{21,c} \mywedge p_{35,c} + p_{35,b} + p_{35,a} = p_{21,c} + p_{21,b} + p_{21,a} ) \end{align*} \label{S_21_35} \tag{$S_{21,35}$}\mypagebreak
\begin{align*}\MoveEqLeft p_{29,a} < p_{22,d} \myvee p_{29,a} + p_{29,c} + p_{29,d} < p_{22,d} + p_{22,b} + p_{22,a} \mybreak\myvee ( p_{29,a} = p_{22,d} \mywedge p_{29,a} + p_{29,c} + p_{29,d} = p_{22,d} + p_{22,b} + p_{22,a} ) \end{align*} \label{S_22_29} \tag{$S_{22,29}$}\mypagebreak
\begin{align*}\MoveEqLeft p_{32,a} < p_{22,a} \myvee p_{32,a} + p_{32,b} < p_{22,a} + p_{22,b} \mybreak\myvee ( p_{32,a} = p_{22,a} \mywedge p_{32,a} + p_{32,b} = p_{22,a} + p_{22,b} ) \end{align*} \label{S_22_32} \tag{$S_{22,32}$}\mypagebreak
\begin{align*}\MoveEqLeft p_{12,c} + p_{12,a} \leq p_{23,c} + p_{23,a} \end{align*} \label{S_23_12} \tag{$S_{23,12}$}\mypagebreak
\begin{align*}\MoveEqLeft p_{18,c} + p_{18,d} \leq p_{23,d} + p_{23,c} \end{align*} \label{S_23_18} \tag{$S_{23,18}$}\mypagebreak
\begin{align*}\MoveEqLeft p_{19,d} + p_{19,b} + p_{19,a} \leq p_{23,d} + p_{23,b} + p_{23,a} \end{align*} \label{S_23_19} \tag{$S_{23,19}$}\mypagebreak
\begin{align*}\MoveEqLeft p_{34,b} < p_{24,c} \myvee p_{34,b} + p_{34,d} < p_{24,c} + p_{24,a} \mybreak\myvee ( p_{34,b} = p_{24,c} \mywedge p_{34,b} + p_{34,d} = p_{24,c} + p_{24,a} ) \end{align*} \label{S_24_34} \tag{$S_{24,34}$}\mypagebreak
\begin{align*}\MoveEqLeft p_{26,d} + p_{26,c} < p_{25,d} + p_{25,b} \myvee p_{26,d} + p_{26,c} + p_{26,a} < p_{25,d} + p_{25,b} + p_{25,a} \mybreak\myvee ( p_{26,d} + p_{26,c} = p_{25,d} + p_{25,b} \mywedge p_{26,d} + p_{26,c} + p_{26,a} = p_{25,d} + p_{25,b} + p_{25,a} ) \end{align*} \label{S_25_26} \tag{$S_{25,26}$}\mypagebreak
\begin{align*}\MoveEqLeft p_{36,a} < p_{25,a} \myvee p_{36,a} + p_{36,c} < p_{25,a} + p_{25,c} \mybreak\myvee ( p_{36,a} = p_{25,a} \mywedge p_{36,a} + p_{36,c} = p_{25,a} + p_{25,c} ) \end{align*} \label{S_25_36} \tag{$S_{25,36}$}\mypagebreak
\begin{align*}\MoveEqLeft p_{8,d} < p_{26,a} \myvee p_{8,d} + p_{8,b} < p_{26,a} + p_{26,c} \mybreak\myvee ( p_{8,d} = p_{26,a} \mywedge p_{8,d} + p_{8,b} = p_{26,a} + p_{26,c} ) \end{align*} \label{S_26_8} \tag{$S_{26,8}$}\mypagebreak
\begin{align*}\MoveEqLeft p_{13,b} + p_{13,a} \leq p_{27,c} + p_{27,a} \end{align*} \label{S_27_13} \tag{$S_{27,13}$}\mypagebreak
\begin{align*}\MoveEqLeft p_{19,d} + p_{19,c} < p_{27,d} + p_{27,b} \myvee p_{19,d} + p_{19,c} + p_{19,a} < p_{27,d} + p_{27,b} + p_{27,a} \mybreak\myvee ( p_{19,d} + p_{19,c} = p_{27,d} + p_{27,b} \mywedge p_{19,d} + p_{19,c} + p_{19,a} = p_{27,d} + p_{27,b} + p_{27,a} ) \end{align*} \label{S_27_19} \tag{$S_{27,19}$}\mypagebreak
\begin{align*}\MoveEqLeft p_{32,d} < p_{28,a} \myvee p_{32,d} + p_{32,b} < p_{28,a} + p_{28,c} \mybreak\myvee ( p_{32,d} = p_{28,a} \mywedge p_{32,d} + p_{32,b} = p_{28,a} + p_{28,c} ) \end{align*} \label{S_28_32} \tag{$S_{28,32}$}\mypagebreak
\begin{align*}\MoveEqLeft p_{39,a} < p_{28,a} \myvee p_{39,a} + p_{39,c} < p_{28,a} + p_{28,b} \mybreak\myvee ( p_{39,a} = p_{28,a} \mywedge p_{39,a} + p_{39,c} = p_{28,a} + p_{28,b} ) \end{align*} \label{S_28_39} \tag{$S_{28,39}$}\mypagebreak
\begin{align*}\MoveEqLeft p_{39,d} < p_{29,a} \myvee p_{39,d} + p_{39,c} < p_{29,a} + p_{29,b} \mybreak\myvee ( p_{39,d} = p_{29,a} \mywedge p_{39,d} + p_{39,c} = p_{29,a} + p_{29,b} ) \end{align*} \label{S_29_39} \tag{$S_{29,39}$}\mypagebreak
\begin{align*}\MoveEqLeft p_{21,b} + p_{21,a} < p_{30,b} + p_{30,a} \myvee p_{21,b} + p_{21,a} + p_{21,d} < p_{30,b} + p_{30,a} + p_{30,d} \mybreak\myvee ( p_{21,b} + p_{21,a} = p_{30,b} + p_{30,a} \mywedge p_{21,b} + p_{21,a} + p_{21,d} = p_{30,b} + p_{30,a} + p_{30,d} ) \end{align*} \label{S_30_21} \tag{$S_{30,21}$}\mypagebreak
\begin{align*}\MoveEqLeft p_{41,c} < p_{30,c} \myvee p_{41,c} + p_{41,b} + p_{41,a} < p_{30,c} + p_{30,b} + p_{30,a} \mybreak\myvee ( p_{41,c} = p_{30,c} \mywedge p_{41,c} + p_{41,b} + p_{41,a} = p_{30,c} + p_{30,b} + p_{30,a} ) \end{align*} \label{S_30_41} \tag{$S_{30,41}$}\mypagebreak
\begin{align*}\MoveEqLeft p_{38,b} + p_{38,d} < p_{31,d} + p_{31,b} \myvee p_{38,b} + p_{38,d} + p_{38,c} < p_{31,d} + p_{31,b} + p_{31,a} \mybreak\myvee ( p_{38,b} + p_{38,d} = p_{31,d} + p_{31,b} \mywedge p_{38,b} + p_{38,d} + p_{38,c} = p_{31,d} + p_{31,b} + p_{31,a} ) \end{align*} \label{S_31_38} \tag{$S_{31,38}$}\mypagebreak
\begin{align*}\MoveEqLeft p_{22,b} + p_{22,a} < p_{32,b} + p_{32,a} \myvee p_{22,b} + p_{22,a} + p_{22,d} < p_{32,b} + p_{32,a} + p_{32,d} \mybreak\myvee ( p_{22,b} + p_{22,a} = p_{32,b} + p_{32,a} \mywedge p_{22,b} + p_{22,a} + p_{22,d} = p_{32,b} + p_{32,a} + p_{32,d} ) \end{align*} \label{S_32_22} \tag{$S_{32,22}$}\mypagebreak
\begin{align*}\MoveEqLeft p_{28,a} < p_{32,d} \myvee p_{28,a} + p_{28,c} + p_{28,d} < p_{32,d} + p_{32,b} + p_{32,a} \mybreak\myvee ( p_{28,a} = p_{32,d} \mywedge p_{28,a} + p_{28,c} + p_{28,d} = p_{32,d} + p_{32,b} + p_{32,a} ) \end{align*} \label{S_32_28} \tag{$S_{32,28}$}\mypagebreak
\begin{align*}\MoveEqLeft p_{5,a} < p_{33,a} \myvee p_{5,a} + p_{5,b} < p_{33,a} + p_{33,b} \mybreak\myvee ( p_{5,a} = p_{33,a} \mywedge p_{5,a} + p_{5,b} = p_{33,a} + p_{33,b} ) \end{align*} \label{S_33_5} \tag{$S_{33,5}$}\mypagebreak
\begin{align*}\MoveEqLeft p_{22,d} + p_{22,c} \leq p_{33,d} + p_{33,c} \end{align*} \label{S_33_22} \tag{$S_{33,22}$}\mypagebreak
\begin{align*}\MoveEqLeft p_{24,c} < p_{34,b} \myvee p_{24,c} + p_{24,a} + p_{24,d} < p_{34,b} + p_{34,d} + p_{34,a} \mybreak\myvee ( p_{24,c} = p_{34,b} \mywedge p_{24,c} + p_{24,a} + p_{24,d} = p_{34,b} + p_{34,d} + p_{34,a} ) \end{align*} \label{S_34_24} \tag{$S_{34,24}$}\mypagebreak
\begin{align*}\MoveEqLeft p_{9,a} < p_{35,a} \myvee p_{9,a} + p_{9,b} < p_{35,a} + p_{35,b} \mybreak\myvee ( p_{9,a} = p_{35,a} \mywedge p_{9,a} + p_{9,b} = p_{35,a} + p_{35,b} ) \end{align*} \label{S_35_9} \tag{$S_{35,9}$}\mypagebreak
\begin{align*}\MoveEqLeft p_{21,c} + p_{21,b} + p_{21,a} \leq p_{35,c} + p_{35,b} + p_{35,a} \end{align*} \label{S_35_21} \tag{$S_{35,21}$}\mypagebreak
\begin{align*}\MoveEqLeft p_{10,d} < p_{36,a} \myvee p_{10,d} + p_{10,c} < p_{36,a} + p_{36,b} \mybreak\myvee ( p_{10,d} = p_{36,a} \mywedge p_{10,d} + p_{10,c} = p_{36,a} + p_{36,b} ) \end{align*} \label{S_36_10} \tag{$S_{36,10}$}\mypagebreak
\begin{align*}\MoveEqLeft p_{25,c} + p_{25,a} < p_{36,c} + p_{36,a} \myvee p_{25,c} + p_{25,a} + p_{25,d} < p_{36,c} + p_{36,a} + p_{36,d} \mybreak\myvee ( p_{25,c} + p_{25,a} = p_{36,c} + p_{36,a} \mywedge p_{25,c} + p_{25,a} + p_{25,d} = p_{36,c} + p_{36,a} + p_{36,d} ) \end{align*} \label{S_36_25} \tag{$S_{36,25}$}\mypagebreak
\begin{align*}\MoveEqLeft p_{39,d} + p_{39,c} \leq p_{36,d} + p_{36,c} \end{align*} \label{S_36_39} \tag{$S_{36,39}$}\mypagebreak
\begin{align*}\MoveEqLeft p_{42,d} < p_{37,a} \myvee p_{42,d} + p_{42,b} < p_{37,a} + p_{37,b} \mybreak\myvee ( p_{42,d} = p_{37,a} \mywedge p_{42,d} + p_{42,b} = p_{37,a} + p_{37,b} ) \end{align*} \label{S_37_42_1} \tag{$S_{37,42}\,(1)$}\mypagebreak
\begin{align*}\MoveEqLeft p_{42,d} < p_{37,c} \myvee p_{42,d} + p_{42,b} < p_{37,c} + p_{37,d} \mybreak\myvee ( p_{42,d} = p_{37,c} \mywedge p_{42,d} + p_{42,b} = p_{37,c} + p_{37,d} ) \end{align*} \label{S_37_42_2} \tag{$S_{37,42}\,(2)$}\mypagebreak
\begin{align*}\MoveEqLeft p_{2,c} + p_{2,a} < p_{39,c} + p_{39,a} \myvee p_{2,c} + p_{2,a} + p_{2,d} < p_{39,c} + p_{39,a} + p_{39,d} \mybreak\myvee ( p_{2,c} + p_{2,a} = p_{39,c} + p_{39,a} \mywedge p_{2,c} + p_{2,a} + p_{2,d} = p_{39,c} + p_{39,a} + p_{39,d} ) \end{align*} \label{S_39_2} \tag{$S_{39,2}$}\mypagebreak
\begin{align*}\MoveEqLeft p_{29,a} + p_{29,b} < p_{39,d} + p_{39,c} \myvee p_{29,a} + p_{29,b} + p_{29,d} < p_{39,d} + p_{39,c} + p_{39,a} \mybreak\myvee ( p_{29,a} + p_{29,b} = p_{39,d} + p_{39,c} \mywedge p_{29,a} + p_{29,b} + p_{29,d} = p_{39,d} + p_{39,c} + p_{39,a} ) \end{align*} \label{S_39_29} \tag{$S_{39,29}$}\mypagebreak
\begin{align*}\MoveEqLeft p_{36,d} + p_{36,c} < p_{39,d} + p_{39,c} \myvee p_{36,d} + p_{36,c} + p_{36,a} < p_{39,d} + p_{39,c} + p_{39,a} \mybreak\myvee ( p_{36,d} + p_{36,c} = p_{39,d} + p_{39,c} \mywedge p_{36,d} + p_{36,c} + p_{36,a} = p_{39,d} + p_{39,c} + p_{39,a} ) \end{align*} \label{S_39_36} \tag{$S_{39,36}$}\mypagebreak
\begin{align*}\MoveEqLeft p_{31,d} + p_{31,b} + p_{31,a} \leq p_{41,c} + p_{41,b} + p_{41,a} \end{align*} \label{S_41_31} \tag{$S_{41,31}$}\mypagebreak
\begin{align*}\MoveEqLeft p_{3,d} < p_{42,d} \myvee p_{3,d} + p_{3,b} < p_{42,d} + p_{42,b} \myvee p_{3,d} + p_{3,b} + p_{3,a} < p_{42,d} + p_{42,b} + p_{42,a} \mybreak\myvee ( p_{3,d} = p_{42,d} \mywedge p_{3,d} + p_{3,b} = p_{42,d} + p_{42,b} \mywedge p_{3,d} + p_{3,b} + p_{3,a} = p_{42,d} + p_{42,b} + p_{42,a} ) \end{align*} \label{S_42_3} \tag{$S_{42,3}$}\mypagebreak
\begin{align*}\MoveEqLeft p_{11,d} < p_{42,c} \myvee p_{11,d} + p_{11,b} < p_{42,c} + p_{42,a} \mybreak\myvee ( p_{11,d} = p_{42,c} \mywedge p_{11,d} + p_{11,b} = p_{42,c} + p_{42,a} ) \end{align*} \label{S_42_11} \tag{$S_{42,11}$}\mypagebreak
\begin{align*}\MoveEqLeft p_{24,b} + p_{24,a} \leq p_{42,b} + p_{42,a} \end{align*} \label{S_42_24} \tag{$S_{42,24}$}\mypagebreak
\begin{align*}\MoveEqLeft p_{12,b} + p_{12,a} < p_{44,b} + p_{44,a} \myvee p_{12,b} + p_{12,a} + p_{12,d} < p_{44,b} + p_{44,a} + p_{44,d} \mybreak\myvee ( p_{12,b} + p_{12,a} = p_{44,b} + p_{44,a} \mywedge p_{12,b} + p_{12,a} + p_{12,d} = p_{44,b} + p_{44,a} + p_{44,d} ) \end{align*} \label{S_44_12} \tag{$S_{44,12}$}\mypagebreak
\begin{align*}\MoveEqLeft p_{40,c} + p_{40,d} \leq p_{44,d} + p_{44,c} \end{align*} \label{S_44_40} \tag{$S_{44,40}$}\mypagebreak
\begin{align*}\MoveEqLeft p_{31,c} + p_{31,d} < p_{45,b} + p_{45,a} \myvee p_{31,c} + p_{31,d} + p_{31,b} < p_{45,b} + p_{45,a} + p_{45,c} \mybreak\myvee ( p_{31,c} + p_{31,d} = p_{45,b} + p_{45,a} \mywedge p_{31,c} + p_{31,d} + p_{31,b} = p_{45,b} + p_{45,a} + p_{45,c} ) \end{align*} \label{S_45_31} \tag{$S_{45,31}$}\mypagebreak
\begin{align*}\MoveEqLeft p_{20,c} + p_{20,a} \leq p_{46,c} + p_{46,a} \end{align*} \label{S_46_20} \tag{$S_{46,20}$}\mypagebreak
\begin{align*}\MoveEqLeft p_{37,a} + p_{37,c} < p_{46,d} + p_{46,b} \myvee p_{37,a} + p_{37,c} + p_{37,d} < p_{46,d} + p_{46,b} + p_{46,a} \mybreak\myvee ( p_{37,a} + p_{37,c} = p_{46,d} + p_{46,b} \mywedge p_{37,a} + p_{37,c} + p_{37,d} = p_{46,d} + p_{46,b} + p_{46,a} ) \end{align*} \label{S_46_37} \tag{$S_{46,37}$}\mypagebreak
\begin{align*}\MoveEqLeft p_{30,b} + p_{30,a} \leq p_{47,b} + p_{47,a} \end{align*} \label{S_47_30} \tag{$S_{47,30}$}\displaybreak[0]\\[-2.2em]\notag
\end{gather}
\noindent \rule{\textwidth}{\heavyrulewidth}
\vspace*{-1em}
\captionof{table}{The strategyproofness conditions used in the impossibility proof.}
\label{tbl:sp}
\end{fleqn}
\endgroup
\vspace{2em}

\Cref{tbl:manipulations} lists the manipulations that were used to obtain these strategyproofness conditions: the first column gives the name of the manipulation condition in the form ($S_{i,j}$), which also contains the information which two profiles are involved in the manipulation ($R_i$ and $R_j$). The next columns contain the manipulating agent, his preferences, and the false preferences that he needs to submit. The last column gives the permutation of the alternatives that yields $R_j$ when applied to the manipulated instance of $R_i$.

\begin{longtable}{ccccc}
\toprule
Condition & Agent & Old Preferences & New Preferences & Permutation\\\midrule
\eqref{S_1_2} & 1 & $\{c,d\},\{a,b\}$ & $\{c,d\},a,b$ & $(a)(b)(c)(d)$\\
\eqref{S_1_19} & 3 & $a,b,\{c,d\}$ & $\{a,b\},\{c,d\}$ & $(a)(b)(c)(d)$\\
\eqref{S_2_1} & 2 & $\{c,d\},a,b$ & $\{c,d\},\{a,b\}$ & $(a)(b)(c)(d)$\\
\eqref{S_2_38} & 1 & $\{a,c\},\{b,d\}$ & $\{a,c\},b,d$ & $(a)(b)(c)(d)$\\
\eqref{S_4_8} & 4 & $d,c,\{a,b\}$ & $c,d,\{a,b\}$ & $(a)(b)(c\ d)$\\
\eqref{S_4_18} & 3 & $c,\{a,b\},d$ & $\{a,b,c\},d$ & $(a)(b)(c)(d)$\\
\eqref{S_4_47} & 2 & $\{a,d\},\{b,c\}$ & $\{a,d\},c,b$ & $(a)(b)(c)(d)$\\
\eqref{S_5_7} & 3 & $\{a,c\},d,b$ & $a,c,d,b$ & $(a)(b)(c)(d)$\\
\eqref{S_5_10} & 4 & $d,\{a,b\},c$ & $\{b,d\},a,c$ & $(a\ d)(b\ c)$\\
\eqref{S_5_17} & 3 & $\{a,c\},d,b$ & $\{a,c\},\{b,d\}$ & $(a)(b)(c)(d)$\\
\eqref{S_6_19} & 4 & $d,b,a,c$ & $\{b,d\},a,c$ & $(a)(b)(c)(d)$\\
\eqref{S_6_42} & 3 & $\{a,c\},\{b,d\}$ & $c,a,\{b,d\}$ & $(a)(b)(c)(d)$\\
\eqref{S_7_43} & 3 & $a,c,d,b$ & $a,\{c,d\},b$ & $(a\ d)(b\ c)$\\
\eqref{S_8_26} & 3 & $d,\{a,b\},c$ & $d,b,\{a,c\}$ & $(a\ d)(b\ c)$\\
\eqref{S_9_18} & 2 & $\{a,d\},c,b$ & $\{a,d\},\{b,c\}$ & $(a)(b)(c)(d)$\\
\eqref{S_9_35} & 1 & $\{a,b\},\{c,d\}$ & $a,b,\{c,d\}$ & $(a)(b)(c)(d)$\\
\eqref{S_9_40} & 1 & $\{a,b\},\{c,d\}$ & $\{a,b\},c,d$ & $(a)(b)(c)(d)$\\
\eqref{S_10_12} & 3 & $\{a,c\},d,b$ & $\{a,c,d\},b$ & $(a\ d)(b\ c)$\\
\eqref{S_10_15} & 4 & $\{b,d\},a,c$ & $d,b,a,c$ & $(a\ d)(b\ c)$\\
\eqref{S_10_19} & 4 & $\{b,d\},a,c$ & $\{b,d\},\{a,c\}$ & $(a\ d)(b\ c)$\\
\eqref{S_10_36} & 2 & $\{c,d\},\{a,b\}$ & $d,c,\{a,b\}$ & $(a\ d)(b\ c)$\\
\eqref{S_12_10} & 4 & $\{a,b,d\},c$ & $\{b,d\},a,c$ & $(a\ d)(b\ c)$\\
\eqref{S_12_16} & 3 & $\{a,c\},d,b$ & $a,c,d,b$ & $(a)(b)(c)(d)$\\
\eqref{S_12_44} & 2 & $\{a,b\},\{c,d\}$ & $\{a,b\},d,c$ & $(a)(b)(c)(d)$\\
\eqref{S_13_15} & 3 & $\{b,d\},a,c$ & $\{b,d\},\{a,c\}$ & $(a)(b\ c)(d)$\\
\eqref{S_13_27} & 4 & $a,b,d,c$ & $\{a,b\},\{c,d\}$ & $(a)(b\ c)(d)$\\
\eqref{S_14_9} & 4 & $a,d,c,b$ & $\{a,d\},c,b$ & $(a)(b)(c)(d)$\\
\eqref{S_14_16} & 2 & $d,c,\{a,b\}$ & $\{c,d\},\{a,b\}$ & $(a)(b)(c\ d)$\\
\eqref{S_14_34} & 3 & $\{a,b,c\},d$ & $b,\{a,c\},d$ & $(a)(b)(c\ d)$\\
\eqref{S_15_5} & 4 & $a,c,d,b$ & $a,\{c,d\},b$ & $(a\ d)(b\ c)$\\
\eqref{S_15_7} & 3 & $\{b,d\},a,c$ & $d,\{a,b\},c$ & $(a)(b)(c)(d)$\\
\eqref{S_15_10} & 4 & $a,c,d,b$ & $\{a,c\},d,b$ & $(a\ d)(b\ c)$\\
\eqref{S_15_13} & 2 & $\{c,d\},\{a,b\}$ & $\{c,d\},a,b$ & $(a)(b\ c)(d)$\\
\eqref{S_17_3} & 3 & $\{a,c\},\{b,d\}$ & $c,a,\{b,d\}$ & $(a)(b)(c)(d)$\\
\eqref{S_17_5} & 3 & $\{a,c\},\{b,d\}$ & $\{a,c\},d,b$ & $(a)(b)(c)(d)$\\
\eqref{S_17_7} & 3 & $\{a,c\},\{b,d\}$ & $a,c,d,b$ & $(a)(b)(c)(d)$\\
\eqref{S_17_11} & 3 & $\{a,c\},\{b,d\}$ & $c,a,b,d$ & $(a)(b)(c)(d)$\\
\eqref{S_18_9} & 2 & $\{a,d\},\{b,c\}$ & $\{a,d\},c,b$ & $(a)(b)(c)(d)$\\
\eqref{S_19_1} & 1 & $\{a,b\},\{c,d\}$ & $a,b,\{c,d\}$ & $(a)(b)(c)(d)$\\
\eqref{S_19_10} & 4 & $\{a,c\},\{b,d\}$ & $\{a,c\},d,b$ & $(a\ d)(b\ c)$\\
\eqref{S_19_27} & 2 & $\{c,d\},\{a,b\}$ & $\{c,d\},a,b$ & $(a)(b\ c)(d)$\\
\eqref{S_20_21} & 3 & $a,c,\{b,d\}$ & $a,\{c,d\},b$ & $(a\ c\ b\ d)$\\
\eqref{S_21_35} & 3 & $c,\{a,b\},d$ & $\{a,b,c\},d$ & $(a)(b)(c)(d)$\\
\eqref{S_22_29} & 3 & $d,\{a,b\},c$ & $\{b,d\},a,c$ & $(a\ d)(b\ c)$\\
\eqref{S_22_32} & 4 & $a,b,\{c,d\}$ & $\{a,b\},d,c$ & $(a)(b)(c)(d)$\\
\eqref{S_23_12} & 3 & $\{a,c\},\{b,d\}$ & $\{a,c\},d,b$ & $(a)(b)(c)(d)$\\
\eqref{S_23_18} & 2 & $\{c,d\},\{a,b\}$ & $c,d,\{a,b\}$ & $(a)(b)(c\ d)$\\
\eqref{S_23_19} & 4 & $\{a,b,d\},c$ & $\{b,d\},a,c$ & $(a)(b)(c)(d)$\\
\eqref{S_24_34} & 3 & $c,a,\{b,d\}$ & $c,\{a,d\},b$ & $(a\ d)(b\ c)$\\
\eqref{S_25_26} & 2 & $\{b,d\},a,c$ & $\{b,d\},\{a,c\}$ & $(a)(b\ c)(d)$\\
\eqref{S_25_36} & 4 & $a,c,\{b,d\}$ & $\{a,c\},d,b$ & $(a)(b)(c)(d)$\\
\eqref{S_26_8} & 4 & $a,c,\{b,d\}$ & $a,\{c,d\},b$ & $(a\ d)(b\ c)$\\
\eqref{S_27_13} & 3 & $\{a,c\},\{b,d\}$ & $a,c,d,b$ & $(a)(b\ c)(d)$\\
\eqref{S_27_19} & 2 & $\{b,d\},a,c$ & $\{b,d\},\{a,c\}$ & $(a)(b\ c)(d)$\\
\eqref{S_28_32} & 4 & $a,c,\{b,d\}$ & $a,\{c,d\},b$ & $(a\ d)(b\ c)$\\
\eqref{S_28_39} & 3 & $a,b,\{c,d\}$ & $\{a,b\},d,c$ & $(a)(b\ c)(d)$\\
\eqref{S_29_39} & 3 & $a,b,\{c,d\}$ & $\{a,b\},d,c$ & $(a\ d)(b\ c)$\\
\eqref{S_30_21} & 4 & $\{a,b\},d,c$ & $a,b,\{c,d\}$ & $(a)(b)(c)(d)$\\
\eqref{S_30_41} & 3 & $c,\{a,b\},d$ & $\{a,b,c\},d$ & $(a)(b)(c)(d)$\\
\eqref{S_31_38} & 1 & $\{b,d\},a,c$ & $\{b,d\},c,a$ & $(a\ c)(b\ d)$\\
\eqref{S_32_22} & 4 & $\{a,b\},d,c$ & $a,b,\{c,d\}$ & $(a)(b)(c)(d)$\\
\eqref{S_32_28} & 3 & $d,\{a,b\},c$ & $d,b,\{a,c\}$ & $(a\ d)(b\ c)$\\
\eqref{S_33_5} & 3 & $a,b,\{c,d\}$ & $\{a,b\},\{c,d\}$ & $(a)(b)(c)(d)$\\
\eqref{S_33_22} & 1 & $\{c,d\},\{a,b\}$ & $d,c,\{a,b\}$ & $(a)(b)(c)(d)$\\
\eqref{S_34_24} & 3 & $b,\{a,d\},c$ & $b,d,\{a,c\}$ & $(a\ d)(b\ c)$\\
\eqref{S_35_9} & 2 & $a,b,\{c,d\}$ & $\{a,b\},\{c,d\}$ & $(a)(b)(c)(d)$\\
\eqref{S_35_21} & 3 & $\{a,b,c\},d$ & $c,\{a,b\},d$ & $(a)(b)(c)(d)$\\
\eqref{S_36_10} & 4 & $a,b,\{c,d\}$ & $\{a,b\},\{c,d\}$ & $(a\ d)(b\ c)$\\
\eqref{S_36_25} & 2 & $\{a,c\},d,b$ & $a,c,\{b,d\}$ & $(a)(b)(c)(d)$\\
\eqref{S_36_39} & 1 & $\{c,d\},\{a,b\}$ & $\{c,d\},a,b$ & $(a)(b)(c)(d)$\\
\eqref{S_37_42_1} & 3 & $a,b,\{c,d\}$ & $a,b,d,c$ & $(a\ d)(b)(c)$\\
\eqref{S_37_42_2} & 4 & $c,d,\{a,b\}$ & $c,d,b,a$ & $(a\ c\ d\ b)$\\
\eqref{S_39_2} & 1 & $\{a,c\},d,b$ & $\{a,c\},\{b,d\}$ & $(a)(b)(c)(d)$\\
\eqref{S_39_29} & 4 & $\{c,d\},a,b$ & $d,c,\{a,b\}$ & $(a\ d)(b\ c)$\\
\eqref{S_39_36} & 4 & $\{c,d\},a,b$ & $\{c,d\},\{a,b\}$ & $(a)(b)(c)(d)$\\
\eqref{S_41_31} & 3 & $\{a,b,c\},d$ & $\{b,c\},a,d$ & $(a)(b)(c\ d)$\\
\eqref{S_42_3} & 3 & $d,b,a,c$ & $d,\{a,b\},c$ & $(a)(b)(c)(d)$\\
\eqref{S_42_11} & 4 & $c,a,\{b,d\}$ & $c,\{a,b\},d$ & $(a\ b)(c\ d)$\\
\eqref{S_42_24} & 2 & $\{a,b\},\{c,d\}$ & $b,a,\{c,d\}$ & $(a)(b)(c)(d)$\\
\eqref{S_44_12} & 3 & $\{a,b\},d,c$ & $\{a,b\},\{c,d\}$ & $(a)(b)(c)(d)$\\
\eqref{S_44_40} & 1 & $\{c,d\},\{a,b\}$ & $c,d,\{a,b\}$ & $(a)(b)(c\ d)$\\
\eqref{S_45_31} & 3 & $\{a,b\},c,d$ & $b,a,\{c,d\}$ & $(a\ d)(b\ c)$\\
\eqref{S_46_20} & 3 & $\{a,c\},\{b,d\}$ & $a,c,\{b,d\}$ & $(a)(b)(c)(d)$\\
\eqref{S_46_37} & 1 & $\{b,d\},a,c$ & $\{b,d\},\{a,c\}$ & $(a\ d)(b\ c)$\\
\eqref{S_47_30} & 1 & $\{a,b\},\{c,d\}$ & $\{a,b\},d,c$ & $(a)(b)(c)(d)$\\

\bottomrule
\caption{The manipulations required to obtain the strategyproofness conditions in \Cref{tbl:sp}.}
\label{tbl:manipulations}
\end{longtable}

\end{appendix}

\end{document}